\documentclass[11pt,lettersize]{article}
\usepackage[utf8]{inputenc}
\usepackage{cite}
\usepackage{fullpage}
\usepackage{pifont}
\usepackage[leftcaption]{sidecap}
\usepackage{makecell}
\usepackage{bm}

\usepackage{standalone}
\usepackage{tikz}
\usetikzlibrary{positioning, arrows.meta}
\usepackage{cite}
\usepackage[table]{xcolor}

\usepackage{amssymb}
\usepackage{amsmath}
\usepackage{amsthm}
\usepackage{graphicx}
\usepackage{multicol}
\usepackage{hyperref}
\hypersetup{colorlinks,allcolors=blue}
\usepackage{thm-restate}

\usepackage[ruled,lined,linesnumbered,noend,noresetcount]{algorithm2e}
\usepackage{mathtools}

\newtheorem{theorem}{Theorem}

\newtheorem{lemma}[theorem]{Lemma}
\newtheorem{proposition}[theorem]{Proposition}

\newcommand{\cmark}{\ding{51}} 
\newcommand{\xmark}{\ding{55}} 

\newcommand{\set}[1]{\left\{#1 \right\}}
\DeclarePairedDelimiter{\parens}{(}{)}
\DeclarePairedDelimiter{\tup}{\langle}{\rangle}

\newcommand{\emptycell}{\mathsf{EMPTY}}
\newcommand{\tombstone}{\mathsf{TOMBSTONE}}
\newcommand{\flaginsert}{\mathsf{tentative}}
\newcommand{\flagmember}{\mathsf{final}}
\newcommand{\flagrestart}{\mathsf{revalidate}}
\newcommand{\flagcell}{\mathsf{marked}}
\newcommand{\deleted}{\mathsf{DELETED}}
\newcommand{\collided}{\mathsf{COLLIDED}}
\newcommand{\abort}{\mathsf{ABORT}}
\newcommand{\CAS}{\mathsf{CAS}}

\newcommand{\LL}{\mathsf{LL}}
\newcommand{\SC}{\mathsf{SC}}
\newcommand{\Read}{\mathsf{Read}}
\newcommand{\Modify}{\mathsf{Modify}}
\newcommand{\insr}{\mathsf{insert}}
\newcommand{\del}{\mathsf{delete}}
\newcommand{\lookup}{\mathsf{lookup}}
\newcommand{\var}[1]{\mathit{#1}}
\newcommand{\func}[1]{\mathsf{#1}}
\newcommand{\ds}[1]{\mathrm{#1}}
\newcommand{\delrel}{~{\leadsto}_{\mathsf{del}}~}
\newcommand{\delrelstar}{~{\leadsto}^{*}_{\mathsf{del}}~}
\newcommand{\mop}{\mathsf{mop}}
\newcommand{\nmop}{\mathsf{nop}}
\newcommand{\pt}{\mathit{pt}}
\newcommand{\lin}{\mathit{lin}}

\newcommand{\tssymbol}{\blacksquare}

\newcommand{\hide}[1]{}

\SetKwProg{Function}{function}{}{end}

\makeatletter
\newcommand{\nosemic}{\renewcommand{\@endalgocfline}{\relax}}
\newcommand{\dosemic}{\renewcommand{\@endalgocfline}{\algocf@endline}}
\let\oldnl\nl
\newcommand{\nonl}{\renewcommand{\nl}{\let\nl\oldnl}}
\makeatother

\title{Space-Efficient Lock-Free Linear-Probing Hash Table}

\author{
  Hagit Attiya\\
  Technion, Israel\\
  \texttt{hagit@cs.technion.ac.il}
  \and
  Rotem Oshman\\
  Tel-Aviv University, Israel\\
  \texttt{roshman@tau.ac.il}
  \and
  Noa Schiller\\
  Tel-Aviv University, Israel\\
  \texttt{noaschiller@mail.tau.ac.il}
}
\date{}

\begin{document}

\maketitle

\begin{abstract}
Linear probing is one of the simplest and most space-efficient approaches to
hash table design, and is widely used in sequential settings due to its compact memory layout.
However, designing a \emph{concurrent} linear-probing hash table with strong liveness guarantees has proved difficult, and only a handful of such algorithms have been proposed, 
all of which either restrict concurrency or rely
on large per-entry metadata, thereby compromising space efficiency.

We present a lock-free linear-probing hash table with wait-free lookups that retains the core advantages of sequential linear probing while handling contention gracefully.
Our design uses only a small amount of metadata per table entry:
a constant number of additional bits when using $\LL/\SC$,
or a logarithmic number of bits when using $\CAS$.
The algorithm is linearizable and lock-free, supports insert, delete, and wait-free lookup operations, and is able to safely reclaim space used by deleted elements without rebuilding the table.

We analyze the amortized step complexity of our hash table assuming no concurrent insertions of the same key, and show that each operation has expected amortized step complexity matching that of sequential linear probing, up to the point contention per key.
\end{abstract}

{
\hypersetup{linkcolor=black}
\small \tableofcontents
}
\clearpage

\section{Introduction}
\label{sec:intro}

\emph{Linear probing} is an approach to hash table design
where the table consists of a contiguous array of memory cells,
and each key $v$ is stored either in its hash location $h(v)$
or, if this location is taken,
in the first free cell following $h(v)$.
To search for key $v$, we start at location $h(v)$
and scan forward, until we either find $v$
or we reach an empty cell.
This simple design avoids pointer indirection and benefits from contiguous memory layout.
When the parameters are tuned properly, a typical operation may access only a single cache line.

While there is a vast literature on \emph{sequential} linear probing hash tables, only a few \emph{concurrent} hash tables are based on linear probing~\cite{PurcellHarris05,GaoGH05,HIHashTableSTOC25,Click2007,MaierSD19},
or indeed, on any form of \emph{open addressing},
where the hash table is one contiguous array, and any key can be stored in any memory location, so that the hash table is ``full'' only when the entire array is full.
This may be due to the fact that implementing a lock-free linear-probing hash table in the concurrent setting is much more challenging than in the sequential setting.
To our knowledge, virtually all concurrent hash tables proposed so far either use pointers and indirection (e.g.,~\cite{HerlihyShavitBook2nd,ShalevShavit06,Michael02,LockFreeCuckoo}),
rely on locks to coordinate concurrent relocation of keys (e.g.,~\cite{HerlihyShavitTzafrir08}),
limit concurrency (e.g.,~\cite{ShunBlellochSPAA14}),
require frequent rebuilding (e.g.,~\cite{GaoGH05,MaierSD19}),
and/or use heavyweight concurrency control mechanisms that make them space-inefficient (e.g.,~\cite{PurcellHarris05}).

Much of the appeal of linear probing comes from its compact, contiguous memory layout:
accessing the hash table requires no indirection and exploits the structure of modern memory hierarchies.
Storing large descriptors, timestamps, or auxiliary pointers in each cell undermines these advantages by increasing space usage and memory traffic.
In this paper we present the first lock-free linear-probing hash table that
uses \emph{bounded per-cell metadata},
and safely reuses memory previously used by deleted elements.
Our hash table enjoys nearly all the benefits of sequential linear probing while handling contention gracefully,
and we prove that its expected amortized step complexity nearly matches that of \emph{sequential} linear probing.

We give two versions of our hash table:
one using the $\LL/\SC$ synchronization primitive
and constant metadata per hash table cell,
and one using $\CAS$ and logarithmic metadata per cell.
Let $U$ be the size of the key domain,
$n$ be the number of processes,
and $m$ be the size of the hash table (i.e., the maximum number of keys that may be in the dictionary simultaneously); 
our main result is the following:


\begin{restatable}{theorem}{MainThm}
\label{thm:main}
There is a lock-free linearizable concurrent linear-probing hash table
with wait-free $\lookup$s,
where each memory cell stores
$\lceil \log U \rceil + O(1)$ bits when using the $\LL/\SC$ primitive; or
$\lceil \log U \rceil + \min\{\lceil \log m \rceil, \lceil \log n \rceil\} + O(1)$ bits when using the $\CAS$ primitive.
\end{restatable}

We analyze the expected amortized runtime of operations in our hash table when there are no concurrent insertions of the same key.
This assumption is used only to simplify the runtime analysis, and is not needed for linearizability or lock-freedom.
We show that in any execution that includes at most $(1-1/x)m$ insertions, with no two insertions of the same key overlapping in time and at most $c$ operations on any given key overlapping in time, the expected amortized runtime of each operation is $O(x^2 + c)$.
This matches the classical analysis of linear probing due to Knuth~\cite{Knuth1963OpenAddressing}, up to an additive \emph{point contention} term of $c$.

To our knowledge, ours is the first lock-free linear-probing hash table that uses bounded memory cells and does not require rebuilding the table, 
as long as the number of
cells occupied by keys
does not exceed the size of the table.
Concurrent insertions of the same key may temporarily take up space in the hash table, but this space is made available for reuse immediately upon the return of the redundant insert operations.

\paragraph{Using tombstones in concurrent linear-probing hash tables.}
A key design choice in implementing a linear-probing hash table is how to handle deletions.
One option is to shift elements to fill the resulting hole; however, shifting is difficult to implement correctly under concurrency, and synchronizing it often requires substantial metadata
(for example, in~\cite{HIHashTableSTOC25}).
Another option is to use \emph{tombstones},
a special symbol that marks
cells previously occupied by a deleted element, 
so that lookups continue traversing past them.
However, tombstones may accumulate and degrade performance unless the table is periodically rebuilt (as in, e.g.,~\cite{GaoGH05,MaierSD19}).

Our hash table uses tombstones, but it allows inserts to \emph{reuse} tombstones, effectively treating them as empty cells where new keys can be inserted.
This introduces a new challenge: concurrent inserts of the same key may observe different tombstones and temporarily create multiple copies of the same key in the table (we give an example of this in Figure~\ref{fig:alg}
below).
Our main algorithmic contribution is a mechanism that resolves such conflicts using only bounded metadata per cell.



In our hash table, 
operations perform steps only ``on their own behalf'' and do not attempt to help other operations to complete.
Rather than ensuring that other inserts of the same key  complete, 
operations are allowed to withdraw and return before ascertaining which copy is the ``final winner''.
This helps us avoid large metadata, but it requires delicate synchronization between concurrent operations on the same key.
For example,
a challenging aspect of our design is that an ultimately \emph{unsuccessful} insert operation may force the key to be treated as present by subsequent operations, even though there has not been a successful insert yet.
The linearizability proof addresses this challenge by introducing the notion of 
 \emph{insert sequences}, defined as chains of insert operations where the tentative copy of one operation is eliminated by the next in the sequence. 
The final copy in the sequence is successful (i.e., returns \textit{true}).
Even with this notion in hand, ordering insert operations is subtle,
because there can be several insert sequences \emph{on the same key} that overlap in time.
In particular, \emph{unsuccessful} insert operations are not necessarily ordered after the ``final winner'' in their own insert sequence.

\section{Additional Discussion of Related Work}
\label{sec:related}

Concurrent hash tables have been studied extensively, with many
designs targeting scalability, non-blocking progress, and dynamic
resizing.
Most modern concurrent hash tables rely on chaining,
bucket-based layouts, or other forms of indirection, often combined
with fine-grained locking or lock-free pointer-based synchronization
(e.g.,~\cite{Michael02,ShalevShavit06,HerlihyShavitBook2nd}).

Only a few concurrent hash tables 
exist that use linear probing, since open addressing severely constrains synchronization.
Table~\ref{tab:lp_comparison} compares  
our implementations and previous work.
Next we briefly discuss the implementations
that appear in the table.

{
    \renewcommand{\arraystretch}{1.5}
\begin{table}[tb]
    \centering
    \small
    \begin{tabular}{| c | c | c | c | c | c |}
        \hline
         & \makecell{\textbf{Memory} \\ \textbf{Primitive}} & \makecell{\textbf{Deletion} \\ \textbf{Method}} & \makecell{\textbf{Cell} \\ \textbf{Reuse?}} &
         \textbf{Cell Size} \\
        \hline
        \cite{PurcellHarris05} & $\CAS$ & Probe Bounding & \cmark
        & Unbounded \\
        \hline
        \cite{GaoGH05,MaierSD19} & $\CAS$ & Tombstones & \xmark 
        & $\lceil \log U + 2 \rceil$ \\
        \hline
        \cite{HIHashTableSTOC25} & {\makecell{{Improper} \\ {$\LL/\SC$}}} & Shifting & \cmark
        & $2 \lceil \log U + 1 \rceil + 2$ \\
        \hline
        \cite{ShunBlellochSPAA14} & $\CAS$ &
        \makecell{Shifting\\ 
        \emph{\footnotesize{(limited concurrency)}}}
            & \cmark 
        & $\lceil \log U+1 \rceil$ \\
        \hline
        \rowcolor{gray!20}
        \textbf{This paper} &  {\setlength{\fboxsep}{0pt}\colorbox{gray!20}{\makecell{         {\textbf{Proper}} \\ {\textbf{$\LL/\SC$}}}}}  & \textbf{Tombstones} & \cmark
        & {\bm{$\lceil \log U + 1 \rceil + 2$}} \\
        \hline
        \rowcolor{gray!20}
        \textbf{This paper} & \textbf{$\CAS$}  & \textbf{Tombstones} & \cmark
        & {\bm{$\lceil \log U + 1 \rceil + \lceil \log n \rceil + 2$}} \\
        \hline
    \end{tabular} 
      
    \vspace{4pt}
    \caption{Comparison of concurrent linear probing implementations, detailing the synchronization primitive used for each one ($\CAS$, improperly-paired $\LL/\SC$, or properly-paired $\LL/\SC$),
    the deletion method,
    whether or not the hash table is able to re-use cells after deletion,
    and the size of each memory cell in terms of the key domain size $U$.
    Cell sizes are reported for the core data structure and do not include auxiliary space used for resizing.}
    \label{tab:lp_comparison}
\end{table}
}

Purcell and Harris~\cite{PurcellHarris05} present one of the earliest
lock-free open-addressing hash tables.
Their algorithm uses quadratic probing and, instead of marking deletions with tombstones, stores an upper bound on the probe length at each cell. 
Their use of quadratic probing appears orthogonal to the synchronization mechanisms used to manage duplicate copies.

Like ours, the algorithm of~\cite{PurcellHarris05} allows multiple copies of the same key to coexist
temporarily; however,
it resolves conflicts using helping together with rich
per-cell metadata.
Duplicate copies are eliminated dynamically based on probe order:
when an insertion encounters an earlier copy, it withdraws its own
copy and helps complete the earlier one.
As a result, insert operations return only after ensuring that some copy has been finalized.
This helping mechanism relies on unbounded timestamps to avoid the
ABA problem.

Removing timestamps completely from the hash table of~\cite{PurcellHarris05} appears difficult
due to the requirements of the helping mechanism.
In particular, because
helping may require tracking and finalizing \emph{multiple} tentative
copies simultaneously,
a na\"ive attempt to replace timestamps by $\LL/\SC$ would result in
\emph{improperly-paired}
$\LL/\SC$ reservations,
where several $\LL$ reservations are simultaneously pending. 
This is incompatible with current architectures that \emph{allow only one active $\LL$ reservation at a time}.
(For example, the ARM architecture explicitly prohibits mismatched addresses between its $\LL/\SC$ equivalents, 
and violating this results in unpredictable behavior~\cite[Section~B2.12.5]{arm_ddi0487_mb}.)
One might also think of using process identifiers to help manage the ABA problem instead of using timestamps,
but this makes it difficult to distinguish between old and new copies created by the same process.

In other lock-free open-addressing implementations that use tombstones to
mark deletions~\cite{GaoGH05,MaierSD19},
tombstones are not reused (or are reused only for the same key), preventing duplicate copies from arising during
overlapping insertions of the same key.
This restriction simplifies synchronization and eliminates the need for per-entry metadata, but causes tombstones to accumulate over time, increasing table occupancy beyond the number of stored keys and requiring periodic rebuilding.
Our construction is the first lock-free linear-probing hash table to combine tombstone reuse with bounded per-cell metadata.
We believe that the dynamic resizing mechanism integrated in these concurrent open-addressing hash tables can be used with our hash table.


As mentioned in Section~\ref{sec:intro},
the alternative to tombstones is to shift keys backwards
to take up the space freed up by deleted keys.
Concurrent linear probing with backward-shift deletion is significantly harder
as it requires updating multiple cells in the hash table.
In~\cite{HIHashTableSTOC25}, a concurrent Robin Hood hash table
uses improperly-paired $\LL/\SC$ to emulate two-word synchronization primitives;  
this hash table also stores two keys in each cell, 
halving its capacity.
Shun and Blelloch~\cite{ShunBlellochSPAA14} implement backward-shift deletion in the
\emph{phase-concurrent} setting, where inserts and deletes are \emph{not allowed} to overlap.

While a $\CAS$-based version of our hash table
can be obtained by using a general-purpose simulation 
to replace $\LL/\SC$ by $\CAS$,
this would incur high memory overhead:
these simulations either increase the size of each memory cell
by at least $n$ bits~\cite{Moir97,IsraeliRappoportPODC94},
or utilize auxiliary memory arrays, thereby reducing data locality~\cite{blellochWei20,AndersonMoir95,JayantiPetrovic03}.
In contrast,
the $\CAS$-based version of our hash table requires
only $\lceil \log n \rceil$ additional bits per cell compared to the $\LL/\SC$ version.

\section{Preliminaries}
\label{sec:prelim}


We use the standard asynchronous shared memory model, in which $n$ processes communicate through shared \emph{memory cells}.
The shared memory is accessed by executing \emph{primitive operations}.
The $\LL$/$\SC$ primitive supports the following operations:
$\LL(x)$ returns the value stored in
memory cell $x$, and $\SC(x, \mathit{new})$ writes the value $\mathit{new}$ to $x$,
if it was not written since the 
last time the process performed an $\LL(x)$ operation (otherwise, $x$ is not modified). $\SC$ returns \textit{true} if it writes successfully, and \textit{false} otherwise.
$\LL/\SC$ operations should be \emph{properly-paired}:
an $\LL(x)$ is eventually followed by $\SC(x,\_)$ on the same memory cell $x$, without an intervening $\LL$ or $\SC$ to a different memory cell.
The $\CAS$ primitive supports the operation $\CAS(x, \mathit{old}, \mathit{new})$,
which atomically checks if the current value of the memory cell $x$ is $\mathit{old}$, and if so, replaces it by $\mathit{new}$ and returns \textit{true};
otherwise, the operation leaves the value unchanged,
and returns \textit{false}.
Memory cells also support read operations.


An implementation of an abstract data type specifies,
for each process, a program for each operation;
when receiving an \emph{invocation} of an operation,
the process takes \emph{steps} according to its program. 
Each step by a process consists of some local computation,
followed by a single primitive operation on the shared memory.
The process may change its local state after a step, and it
may return a \emph{response} to the operation.

A \emph{configuration} $C$ specifies the state of every process and every memory cell.
An \emph{execution} $\alpha$ is an alternating (finite or infinite) sequence of configurations and steps, starting with a configuration. 
Operation $o_1$ \emph{precedes} operation $o_2$ in $\alpha$, if $o_1$'s response precedes $o_2$'s invocation in $\alpha$; in this case, we also say 
that $o_2$ \emph{follows} $o_1$.

An execution $\alpha$ induces a \emph{history} $H(\alpha)$, 
consisting only of the invocations and responses of higher-level operations.
An invocation \emph{matches} a response if they both belong to the same operation.
An operation \emph{completes} in $H$
if $H$ includes both the invocation and response of the operation;
if $H$ includes the invocation of an operation, but no matching response, 
then the operation is \emph{pending}.
A history $H$ is \emph{sequential} if every invocation is immediately 
followed by a matching response.

A \emph{completion} of history $H$ is a history $H'$ whose prefix 
is identical to $H$, and whose suffix includes zero or more responses 
of pending operations in $H$. 
Let $\mathrm{comp}(H)$ be the set of all of $H$'s completions. 
A sequential history $H'$ is a \emph{linearization} of an execution $\alpha$ if (1) $H'$ is a permutation of a history in $\mathrm{comp}(H(\alpha))$, (2) $H'$ matches the sequential specification of the implemented object, and (3) $H'$ respects the real-time order of non-overlapping operations in $H(\alpha)$.
An execution $\alpha$ is \emph{linearizable} if it has a linearization, 
and an implementation of an abstract object is linearizable 
if all of its executions are linearizable.

An implementation is \emph{lock-free} if whenever there is a pending operation, 
 some operation returns in a finite number of steps of all processes; it is \emph{wait-free} if whenever there is a pending operation by process $p$, 
this operation returns in a finite number of steps by $p$.
\section{The Hash Table Construction} 
\label{sec:alg}


A hash table implements a \emph{dictionary}, 
which represents an unordered set of elements,
and supports the operations 
$\insr(v)$, $\del(v)$ and $\lookup(v)$.
Each operation takes a \emph{key} $v$ and returns \emph{true} or \emph{false}:
$\lookup(v)$ and $\del(v)$ return \emph{true}
if $v$ is in the set,
while $\insr(v)$ returns \emph{true} if $v$ is not in the set.
An $\insr$ may return a special value, $\abort$, indicating that no available slot was found to insert the key (i.e., the hash table is ``full''). 
Aborts are expected to be a rare event and they should trigger a rebuild of the hash table using standard techniques, which we do not treat here
 (see Section~\ref{sec:resize} for further discussion).
$\abort$s do not affect the logical state of the hash table;
in terms of the sequential semantics, we model them by allowing $\insr$ operations to nondeterministically return $\abort$ at any point, without modifying the table.


\subsection{High-Level Overview}
\label{sec:alg overview}

The main challenge in a concurrent linear-probing hash table that reuses tombstones
is handling overlapping insertions of the same key,
because such insertions may temporarily create multiple copies of the key in the table.
We must handle these copies in a way that ensures that
\emph{exactly one copy} eventually survives (barring deletion of the key). 

The key invariant underlying the algorithm is that surviving copies of a key
can move only toward earlier positions in the probe sequence beginning at $h(v)$.
Intuitively, among multiple copies of a key $v$,
copies that are closer to $h(v)$ have priority over later copies.
This probe-order priority prevents cycles of competing eliminations.
A similar approach was taken in~\cite{PurcellHarris05},
but here we avoid helping and large metadata.

To realize the invariant,
an insert operation first places a \emph{tentative copy} of the key in the table.
Tentative copies may subsequently be finalized, withdrawn, or eliminated by competing operations.
Operations searching for the key treat tentative copies as logically present,
even before the surviving copy is determined.
This behavior is essential for linearizability, and is formalized later using the notion of \emph{insert sequences} (see Section~\ref{sec:insert-seq}).

We focus first on the version of the algorithm that uses the $\LL/\SC$ primitive,
in which each memory cell stores a key together with only a constant number of tag bits. 
At the end of Section~\ref{sec:code} we show how to modify the hash table to use $\CAS$, by storing an index of a cell or a process within each cell (an additional $O(\min(\log m, \log n))$ bits).

\paragraph{Using metadata to handle duplicates.}

We use constant-size metadata to manage the competition between overlapping $\insr$s, as well as the interaction with concurrent $\del$ and $\lookup$ operations on the same key.
A key $v$ is first inserted into the table as a \emph{tentative copy},
$\langle v, \flaginsert \rangle$.
The tentative copy reserves the cell for the inserting operation, although other operations may still modify the metadata in the cell in order to communicate with the inserting process.
There are several ways the insertion can proceed from here (see Figure~\ref{fig:life-cycle}):

\sidecaptionvpos{figure}{t}
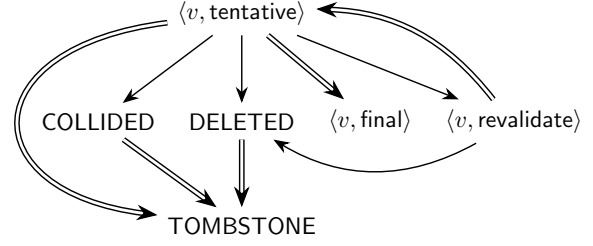
\begin{SCfigure}[40][tb] 
    \centering
    \begin{tikzpicture}[
    >={Stealth[scale=1.2]},
    node distance=1.5cm and 2.5cm,
    every node/.style={font=\small, inner sep=4pt},
    connection/.style={->, line width=0.6pt}
]

    \node (start) {$\tup{v,\flaginsert}$};

    \node (del) [below=1cm of start] {$\deleted$};
    \node (col) [left=0.2cm of del] {$\collided$};

    \node (mem)  [right=0.2cm of del] {$\tup{v,\flagmember}$};

    \node (restart)  [right=0.2cm of mem] {$\tup{v,\flagrestart}$};
    
    \node (tomb) [below=1cm of del] {$\tombstone$};

    \draw[connection] (start) -- (col);
    \draw[connection] (start) -- (del);
    \draw[connection] (start) -- (restart);
    \draw[connection, double, double distance=1.2pt] (start) -- (mem);

    \draw[connection, double, double distance=1.2pt] (col) -- (tomb);
    \draw[connection, double, double distance=1.2pt] (del) -- (tomb);

    \draw[connection] (restart) to[bend left=30] (del);

    \draw[connection, double, double distance=1.2pt] (restart) to[bend right=25] (start);

    \draw[connection, double, double distance=1.2pt] (start) .. controls +(-4, -0.5) and +(-4, 0.5) .. (tomb);

\end{tikzpicture}
    \caption{The life cycle of a tentative copy of key $v$. Double lines indicate changes that can only be made by the ``owner'', i.e.,
    the insert operation that created the tentative copy.
    In addition,
    the value $\collided$ is  written only by a concurrent $\insr(v)$ operation,
    and $\deleted$ only by a concurrent $\del(v)$ operation.}
    \label{fig:life-cycle}
\end{SCfigure}

\begin{itemize}
    \item The insertion can be finalized by having the original insert operation write  $\langle v, \flagmember \rangle$.
    \item The tentative copy can be withdrawn by writing $\tombstone$ into the cell. Only the original inserting operation can do this, either upon request from another operation (described next) or upon finding another copy of the same key closer to location $h(v)$.
    \item A different $\insr(v)$ or $\del(v)$ operation can \emph{eliminate} a copy of $v$ by writing $\collided$ or $\deleted$, 
    respectively. This will cause the original inserting operation to clean up the cell and convert it into a $\tombstone$.
    The value written to the cell determines the response of the inserting operation: \textit{true} for $\deleted$ and \textit{false} for $\collided$.
    \item A different process carrying out $\lookup(v)$ or $\insr(v)$ can also signal the inserting operation 
    that it has \emph{seen} the tentative copy and ``relied on it'',
    by writing    
    $\tup{v, \flagrestart}$.
    This causes the original inserting operation to
    stop trying to withdraw its tentative copy if it is currently doing so,
    and to revalidate it
    by checking again for duplicates (see below).
\end{itemize}

\paragraph{Validating a tentative copy by checking for duplicates.}
After writing a tentative copy of its key into some location $i$, an $\insr$ operation re-scans from location $h(v)$ to the end of the run
(the maximal contiguous sequence of non-$\emptycell$ cells starting from $h(v)$),
searching for other copies of $v$.
If a competing copy of $v$ is found \emph{before} location $i$ (i.e., between $h(v)$ and $i$),
then the other copy ``wins'', and the insert operation attempts to withdraw its tentative copy and convert it into a tombstone.
This may fail if another process signals that it has seen the tentative copy, causing the operation to restart the validation process.
If a competing tentative copy of $v$ is found in some location $j$ \emph{after} location $i$,
then the copy in the earlier location $i < j$ ``wins'', and the inserting process \emph{eliminates} the copy in location $j$ by writing $\collided$ into location $j$ (see more on this process below).
This, too, may fail --- if the copy in location $j$ is already finalized before the value $\collided$ can be written, then location $j$ becomes the final copy,
and the $\insr$ operation for location $i$ will give up and try to delete its own tentative copy.

Although we refer to them as \emph{tentative copies},
if there is a tentative copy of key $v$ in the table,
all later operations must behave as though $v$ is logically present in the table.
The ``tentative'' aspect is only regarding \emph{which copy will win}, not regarding \emph{whether $v$ is inserted}.


Because each tentative copy can only (try to) eliminate other tentative copies that precede it in the hash table, 
a chain of attempted inserts of the same key $v$ that eliminate each other goes
\emph{backwards} through the hash table, towards location $h(v)$.
We refer to such chains as \emph{insert sequences},
and the fact that surviving copies can only move ``up the table'' towards $h(v)$ is crucial to the correctness of the algorithm.
Figure~\ref{fig:alg} illustrates how a copy moves towards $h(v)$.

\paragraph{Safely eliminating tentative copies.}
A critical part of the algorithm is the safe elimination of later tentative copies of the same key.
Before eliminating another copy of $v$,
an $\insr(v)$ operation must ensure that its own
tentative copy has not been concurrently deleted:
without this check, an $\insr$ that already ``lost'' (i.e., been deleted by a concurrent $\del(v)$ operation,
or had its copy eliminated by a concurrent $\insr(v)$)
could erroneously eliminate 
tentative copies belonging to future insertions that began after it had already lost.
This could lead to \emph{no copies of $v$ surviving}, even though ``logically'', $v$ has been re-inserted into the table after being deleted.


\paragraph{Why revalidation is required.}
As we said above, a $\lookup(v)$ or $\insr(v)$ operation that encounters a tentative copy of $v$ 
will signal the inserting operation to revalidate, i.e., stop withdrawing the copy if it is currently doing so, and check again for duplicates.
Subsequently, the $\lookup$ or $\insr$ operation treats the tentative copy of $v$ as though it were finalized --- $\lookup(v)$ returns \emph{true}, and $\insr(v)$ returns \emph{false}.

Requesting revalidation is necessary because the $\lookup$ or $\insr$ operation has no way of knowing whether the tentative copy that it has seen is about to be withdrawn.
For example, suppose operation $I = \insr(v)$ has already written its tentative copy into location $i$,
but then found another tentative copy in location $h(v) < j < i$
and decided to withdraw its copy in location $i$.
Before $I$ can write a tombstone to location $i$,
the preceding copy of $v$ in location $j$ is deleted.
Then a $\lookup(v)$ or $\insr(v)$ operation $O$ begins scanning forward from $h(v)$ and finds the copy of $v$ in location $i$, even though the key is no longer logically in the set.
Operation $O$ has no way of knowing whether operation $I$ is about to finalize its copy or withdraw it, and does not wait around to see;
instead, it treats the key as though it is bound to be finalized, and signals to $I$ that it has done so by requesting revalidation.
This ``resurrects'' the key in location $i$ and causes $I$ to resume trying to finalize it.

\sidecaptionvpos{figure}{t}
\begin{SCfigure}[50][tb] 
\centering
  \begin{tabular}{l@{}l@{}} 

    (a) \hspace{0.2cm} &
    \begin{tabular}{|p{3ex}|p{2.5ex}|p{2.5ex}|p{2.5ex}|p{2.5ex}|p{2.5ex}|p{2.5ex}|p{3ex}|}
    \multicolumn{1}{c}{} & \multicolumn{1}{c}{} & \multicolumn{1}{c}{} & \multicolumn{1}{c}{} & \multicolumn{1}{c}{} & \multicolumn{1}{c}{} & \multicolumn{1}{c}{$I_1$} & \multicolumn{1}{c}{} \\
    \hline
    \centering $\ldots$ & \centering $\ast$ &
     \centering $\ast$ &  \centering $\ast$ &  \centering $\tssymbol$ &  \centering $\ast$ &   & 
    {\centering $\ldots$} \tabularnewline
    \hline
    \end{tabular} 
    \\[2.5ex]

    (b) \hspace{0.2cm} &
    \begin{tabular}{|p{3ex}|p{2.5ex}|p{2.5ex}|p{2.5ex}|p{2.5ex}|p{2.5ex}|p{2.5ex}|p{3ex}|}
    \multicolumn{1}{c}{} & \multicolumn{1}{c}{} & \multicolumn{1}{c}{} & \multicolumn{1}{c}{} & \multicolumn{1}{c}{$I_2$} & \multicolumn{1}{c}{} & \multicolumn{1}{c}{$I_1$} & \multicolumn{1}{c}{} \\
    \hline
    \centering $\ldots$ & \centering $\tssymbol$ &
     \centering $\ast$ &  \centering $\ast$ &  \centering $\tssymbol$ &  \centering $\ast$ &   & 
    {\centering $\ldots$} \tabularnewline
    \hline
    \end{tabular} 
    \\[2.5ex]

    (c) \hspace{0.2cm} &
    \begin{tabular}{|p{3ex}|p{2.5ex}|p{2.5ex}|p{2.5ex}|p{2.5ex}|p{2.5ex}|p{2.5ex}|p{3ex}|}
    \multicolumn{1}{c}{} & \multicolumn{1}{c}{$I_3$} & \multicolumn{1}{c}{} & \multicolumn{1}{c}{} & \multicolumn{1}{c}{$I_2$} & \multicolumn{1}{c}{} & \multicolumn{1}{c}{$I_1$} & \multicolumn{1}{c}{} \\
    \hline
    \centering $\ldots$ & \centering $v?$ &
     \centering $\ast$ &  \centering $\ast$ &  \centering $v?$ &  \centering $\ast$ &  \centering $v?$  & 
    {\centering $\ldots$} \tabularnewline
    \hline
    \end{tabular} 
    \\[2.5ex]

    (d) \hspace{0.2cm} &
    \begin{tabular}{|p{3ex}|p{2.5ex}|p{2.5ex}|p{2.5ex}|p{2.5ex}|p{2.5ex}|p{2.5ex}|p{3ex}|}
    \multicolumn{1}{c}{} & \multicolumn{1}{c}{$I_3$} & \multicolumn{1}{c}{} & \multicolumn{1}{c}{} & \multicolumn{1}{c}{$I_2$} & \multicolumn{1}{c}{} & \multicolumn{1}{c}{} & \multicolumn{1}{c}{} \\
    \hline
    \centering $\ldots$ & \centering $v?$ &
     \centering $\ast$ &  \centering $\ast$ &  \centering {\tiny $\mathsf{COL}$} &  \centering $\ast$ &  \centering $\tssymbol$  & 
    {\centering $\ldots$} \tabularnewline
    \hline
    \end{tabular} 
    \\[2.5ex]

    (e) \hspace{0.2cm} &
    \begin{tabular}{|p{3ex}|p{2.5ex}|p{2.5ex}|p{2.5ex}|p{2.5ex}|p{2.5ex}|p{2.5ex}|p{3ex}|}
    \multicolumn{1}{c}{} & \multicolumn{1}{c}{$I_3$} & \multicolumn{1}{c}{} & \multicolumn{1}{c}{} & \multicolumn{1}{c}{} & \multicolumn{1}{c}{} & \multicolumn{1}{c}{} & \multicolumn{1}{c}{} \\
    \hline
    \centering $\ldots$ & \centering $v$ &
     \centering $\ast$ &  \centering $\ast$ &  \centering $\tssymbol$ &  \centering $\ast$ &  \centering $\tssymbol$  & 
    {\centering $\ldots$} \tabularnewline
    \hline
    \end{tabular} 
    \\[2.5ex]

  \end{tabular}

  \caption{ Three concurrent insertions, $I_1$, $I_2$ and $I_3$, of the same key $v$; cells occupied by keys other than $v$ hold $\ast$.
  (a) $I_1$ prepares to write a tentative copy; deletion creates a tombstone.
  (b) $I_2$ finds a tombstone and prepares to write a tentative copy; deletion creates another tombstone.
  (c) $I_3$ finds a tombstone; all three operations write  tentative copies ($v?$).
  (d) $I_1$ sees copy written by $I_3$ and withdraws; $I_3$ sees copy written by $I_2$ and eliminates it by writing $\collided$.
  (e) $I_2$ cleans up by writing a tombstone; $I_3$ promotes its copy to final.
  }
  \label{fig:alg}
\end{SCfigure}

\subsection{Detailed Description}
\label{sec:code}

We now describe how these ideas are realized in the code.
We provide two versions of the algorithm. 
The first uses the $\LL$/$\SC$ primitive to safely verify and delete tentative copies. 
The second uses the more common $\CAS$ primitive, 
emulating the required behavior with a small amount of additional
per-cell metadata. Both implementations have the same high-level structure, and we encapsulate the use of the two different synchronization primitives by using a few primitive-specific functions (Algorithm~\ref{alg:ll_sc_code} for $\LL$/$\SC$ and Algorithm~\ref{alg:cas_code} for $\CAS$) and by using the keywords $\Read$ and $\Modify$ in the code.
When using $\LL$/$\SC$,
the keyword $\Read$ is an invocation of $\LL$,
and $\Modify$ is an $\SC$;
in the $\CAS$ version,
$\Read$ is standard atomic read,
and $\Modify$ is an invocation of $\CAS$.

We use modular arithmetic that wraps at the $m$-th cell, and omit the notation $(\mathrm{mod} \medspace m)$.

\paragraph{Data structures.}
The table is a single array $\ds{table}[0 \dots m-1]$,
where each cell stores either a tagged key $\langle v, \mathit{tag} \rangle$ where $\mathit{tag} \in \set{ \flaginsert, \flagmember, \flagrestart }$, 
or one of the values 
$\emptycell$, $\tombstone$ representing an available cell,
or one of a few internal cleanup states
($\deleted$, $\collided$, and in the $\CAS$ version,
$\tup{\tup{v,j},\flagcell}$, where $j$ is the index of a cell).
Using the two tag bits, one reserved sentinel value is sufficient to encode the four states that do not contain a key. Thus,
these values can be encoded using $\lceil \log (U+1) \rceil + 2$ bits
when using $\LL$/$\SC$, and $\lceil \log(U+1) \rceil + \min(\lceil \log n \rceil, \lceil \log m \rceil) + 2$ bits for $\CAS$.
In the code, the function 
$\func{val}(x)$ extracts the key from encoded cell $x$ (or $\bot$), 
and $\func{restart}(x)$ checks whether the owner of a tentative copy 
should validate it, which holds if $x = \tup{v,\flagrestart}$ or $x = \tup{\tup{v,\ast},\flagcell}$ (with $\CAS$). 

\paragraph{Searching for a key.}
All operations begin by 
searching for an existing copy of key $v$.
First, operations execute a \emph{forward scan},
which begins at location $h(v)$ and scans forward until it either finds $v$ in the hash table or reaches the end of the run (an $\emptycell$ cell).
If the key was not found during the forward scan,
operations then execute a \emph{backward scan} from the end of the run back to $h(v)$.
Since a sequence of tentative copies that eliminate one another moves back through the table towards location $h(v)$,
a backward scan ensures that at least one of these copies is eventually encountered, provided at least one remains in the table.

At any point during the forward and the backward scan, if key $v$ is found,
then action may be taken, depending on the type of the operation:
a $\del(v)$ operation will try to delete the key, and
a $\lookup(v)$ or $\insr(v)$ operation return \emph{true} or \emph{false},
respectively.
We remark that ``finding key $v$'' does not merely amount to finding a cell where a copy of key $v$ is stored:
if the copy that is found is not finalized (i.e., not $\langle v, \flagmember \rangle$) 
and the searching operation is $\lookup(v)$ or $\insr(v)$,
then the searching operation will request that the original inserting operation  revalidate its copy, by writing $\langle v, \flagrestart \rangle$ into the cell (if no other process already requested this).

Operations $\del(v)$ and $\lookup(v)$ return after completing their forward and backward scan, but for $\insr(v)$ the work has only begun.
At the end of the backward scan, upon returning to location $h(v)$,
an $\insr(v)$ operation probes forward until it finds an available cell, either $\emptycell$ or $\tombstone$.
It then attempts to insert a tentative copy of key $v$ by writing $\langle v, \flaginsert \rangle$, and if successful, checks for duplicates (as explained above).

\paragraph{Scanning.}
The forward scan and the backward scan are detailed in Algorithm~\ref{alg:scan}.
The forward scan from $h(v)$ proceeds until encountering $\emptycell$
or completing a full cycle.
At each cell that contains the key $v$,
the scan invokes either
$\func{validate\_copy}$ (for insert/lookup) or
$\func{try\_delete}$ (for delete);
these helper functions are described below,
but intuitively, $\func{validate\_copy}$ attempts to confirm the presence of a ``stable copy'' of the key, and returns \emph{true} if it succeeds;
$\func{try\_delete}$ attempts to convert the key into a tombstone or signal the inserting process to do so, returning a non-$\bot$ value if it succeeds.
If the forward scan sees no decisive result, it returns the index of the last non-empty cell in $v$'s probe run.
The backward scan then moves backwards toward $h(v)$ and searches for $v$, again invoking the
appropriate helper functions. 

\begin{algorithm}[tb]
{\small
    \DontPrintSemicolon
        \Function{$\func{forward\_scan}(v, \var{op\_type})$:}{
        $i \gets h(v)$ ;
        $\var{val} \gets \Read(\ds{table}[i])$ \;
        \While{$\var{val} \neq \emptycell$}{
        \If{$\func{val}(\var{val}) = v$}{
            \lIf{$\var{op\_type} \in \set{\insr, \lookup}$ and $\func{validate\_copy}(v,\var{val},i)$}{
                \Return \textit{true}
            }
            \If{$\var{op\_type} = \del$}{
                $\var{rsp} \gets \func{try\_delete}(v,\var{val},i)$ \;
                \lIf{$\var{rsp} \neq \bot$}{
                    \Return $\var{rsp}$
                }
            }}
            $i \gets i + 1$ \;
            \lIf{$\var{i} = h(\var{v})$}{
                break
                \tcp*[f]{Scanned all cells in the table}
            }
            $\var{val} \gets \Read(\ds{table}[i])$ \;
        }
         \lIf{$i = h(v)$ and $\var{val} = \emptycell$}{
            \Return $i$
            \tcp*[f]{An empty run}
        }
        \lElse(\tcp*[f]{Go back one cell to the end of the run}){\Return $i - 1$}
        }
        \BlankLine
        \Function{$\func{backward\_scan}(v, i, \var{op\_type})$:}{
        \While{\textit{true}}{
            $\var{val} \gets \Read(\ds{table}[i])$ \;
            \If{$\func{val}(\var{val}) = v$}{
            \lIf{$\var{op\_type} \in \set{\insr, \lookup}$ and $\func{validate\_copy}(v,\var{val},i)$}{
                \Return \textit{true}
            }
            \If{$\var{op\_type} = \del$}{
                $\var{rsp} \gets \func{try\_delete}(v,\var{val},i)$ \;
                \lIf{$\var{rsp} \neq \bot$}{
                    \Return $\var{rsp}$
                }
            }}
            \lIf{$i = h(v)$}{
                \Return $\bot$
                \tcp*[f]{Back to the start}
            }
            $i \gets i - 1$
        }
        }
    \caption{forward and backward scan}
    \label{alg:scan}
    }
\end{algorithm}

\paragraph{Lookup and delete.}
Lookup and delete operations (Algorithm~\ref{alg:deletelookup})
both perform a forward scan and, 
if needed, a backward scan.
Upon finding a copy of $v$, they invoke a helper function ($\func{validate\_copy}$ or $\func{try\_delete}$, respectively; see Algorithm~\ref{alg:helper} below),
and if that helper function returns \textit{true}, the operation returns \textit{true}.

\begin{algorithm}[tb]
    \LinesNumbered
    \small
    \begin{minipage}[t]{0.46\textwidth}
        \DontPrintSemicolon
        \Function{$\del(v)$:}{
        \nl $\var{rsp} \gets \func{forward\_scan}(v, \del)$ \;
        \nl \lIf{$\var{rsp} \in \set{\textit{true}, \textit{false}}$}{\Return $\var{rsp}$}
        \nl $\var{rsp} \gets \func{backward\_scan}(v, \var{rsp}, \del)$ \;
        \nl \lIf{$\var{rsp} \in \set{\textit{true}, \textit{false}}$}{\Return $\var{rsp}$}
        \nl \Return \textit{false} \;
        }
    \end{minipage}
    \hfill
    \begin{minipage}[t]{0.43\textwidth}
        \DontPrintSemicolon
        \vspace{-0.63\baselineskip} \nl \Function{$\lookup(v)$:}{
        \nl $\var{rsp} \gets \func{forward\_scan}(v, \lookup)$ \;
        \nl \lIf{$\var{rsp} = \textit{true}$}{\Return \textit{true}}
        \nl $\var{rsp} \gets \func{backward\_scan}(v, \var{rsp}, \lookup)$ \;
        \nl \lIf{$\var{rsp} = \textit{true}$}{\Return \textit{true}}
        \nl \Return \textit{false}
        }
    \end{minipage}
    \caption{$\del$ and $\lookup$}
    \label{alg:deletelookup}
\end{algorithm}

\paragraph{Insert.}
An insert (Algorithm~\ref{alg:insert}) first checks for an existing copy via the forward and backward scans; if found (using the same criteria as a successful lookup),
it returns \textit{false}.
Otherwise, it probes for an $\emptycell$/$\tombstone$ cell and
attempts to write $\tup{v,\flaginsert}$.
It then rescans the run of~$v$: encountering an earlier or already inserted
copy causes the operation to delete its own tentative copy using $\func{del\_copy}$ and abort.
If $\func{del\_copy}$ returns $\bot$, this means another operation requested revalidation, so the key must be ``resurrected'' and revalidated.
Encountering a later tentative copy triggers $\func{del\_other\_copy}$,
which attempts to eliminate it (returning failure if this attempt is unsuccessful).
This attempt can fail if the later copy was already promoted to final, if the eliminating process's own tentative copy has been marked as deleted, or if another process has requested revalidation.
Upon reaching $\emptycell$, the operation tries to promote its tentative
copy to $\tup{v,\flagmember}$.
If its tentative copy has been marked $\flagrestart$, it repeats the
validation; if marked $\collided$ or $\deleted$, it cleans up and returns \textit{false} or \textit{true}, respectively.

\begin{algorithm}[tb]
{\small
        \DontPrintSemicolon
        \Function{$\insr(v)$:}{
        $\var{rsp} \gets \func{forward\_scan}(v, \insr)$ \;\label{lin:insert-start}
        \lIf{$\var{rsp} = \textit{true}$}{
            \Return \textit{false}
        }
        $\var{rsp} \gets \func{backward\_scan}(v, \var{rsp}, \insr)$ \;
        \lIf{$\var{rsp} = \textit{true}$}{
            \Return \textit{false}
        }
        $j \gets h(v)$ ;
        $\var{val} \gets \Read(\ds{table}[j])$ \;
        \While(\tcp*[f]{Try to insert $v$}){\textit{true}}{
            \If{$\var{val} \in \set{\emptycell,\tombstone}$ and $\Modify(\ds{table}[j], \var{val}, \tup{v,\flaginsert})$}{
            break
        }
        $\var{j} \gets \var{j} + 1$ \;
        \lIf(\tcp*[f]{Scanned all cells in the table}){$j = h(v)$}{
            \Return $\abort$
            \label{line:abort}
        }
        $\var{val} \gets \Read(\ds{table}[j])$ \;
        }
        \SetNoFillComment
        \BlankLine
        \tcp{search for duplicate copies}
        $i \gets h(v)$ ; 
        $\var{val} \gets \Read(\ds{table}[i])$
                \label{lin:start-search}
        \; 
        \While{$\var{val} \neq \emptycell$}{
            \If(\tcp*[f]{Skip own copy and cells without $v$}){$i\neq j$ and $\func{val}(\var{val}) = v$}{
            \If(\tcp*[f]{Other copy is earlier or final}){$i -  h(\var{v}) < j - h(\var{v})$ or $\var{val} = \tup{v, \flagmember}$}{
                $\var{rsp} \gets \func{del\_copy}(v, j)$ \;
                \lIf(\tcp*[f]{Revalidate copy}){$\var{rsp} = \bot$}{
                    goto Line~\ref{lin:start-search}
                }
                \Return $\var{rsp}$ \;
            }
            \ElseIf{$\var{val} \neq \tup{v,\flagrestart}$}{
                \If(\tcp*[f]{Try to eliminate other copy}){$\neg \func{del\_other\_copy}(v, \var{val}, i, j)$}{
                    $\var{cur\_val} \gets \Read(\ds{table}[j])$ \;
                    \If{$\func{restart}(\mkern-1mu\var{cur\_val}\mkern1mu)$ and $\Modify(\ds{table}[j], \var{cur\_val}, \tup{v,\flaginsert})$}{
                        goto Line~\ref{lin:start-search}
                        \tcp*[f]{Revalidate copy}
                    }
                    $\var{rsp} \gets \func{del\_copy}(v, j)$ \;
                    \lIf(\tcp*[f]{Revalidate copy}){$\var{rsp} = \bot$}{
                        goto Line~\ref{lin:start-search}
                    }
                    \Return $\var{rsp}$ \;
                }
             }
            }
            $i \gets i + 1$ \;
            \lIf{$i = h(v)$}{
                break
                \tcp*[f]{Scanned all cells in the table}
            }
            $\var{val} \gets \Read(\ds{table}[i])$ \;
        }
        $\var{cur\_val} \gets \Read(\ds{table}[j])$ \;
        \If{$\var{cur\_val} = \tup{v,\flaginsert}$ and $\Modify(\ds{table}[j], \var{cur\_val}, \tup{v,\flagmember})$}{
            \tcp{Try to fully insert own copy}
            \Return \textit{true}
        }
        \If{$\func{restart}(\mkern-1mu\var{cur\_val}\mkern1mu)$ and $\Modify(\ds{table}[j], \var{cur\_val}, \tup{v,\flaginsert})$}{
                goto Line~\ref{lin:start-search}
                \tcp*[f]{Revalidate copy}
            }
        $\var{rsp} \gets \func{del\_copy}(v, j)$ \;
        \lIf(\tcp*[f]{Revalidate copy}){$\var{rsp} = \bot$}{
            goto Line~\ref{lin:start-search}
        }
        \Return $\var{rsp}$ \;
        }
    \caption{$\insr$}
    \label{alg:insert}
    }
\end{algorithm}

\paragraph{Helper functions.}
The function $\func{validate\_copy}$
(used by $\lookup$ and $\insr$) returns \textit{true} if it identifies a finalized copy of $v$, marks a copy of $v$ for revalidation (or sees an existing request), or observes that a copy of $v$ was rewritten; it returns \textit{false} if re-reading the cell reveals that the copy of $v$ is no longer present.
If neither scan finds a copy of $v$ that is confirmed by $\func{validate\_copy}$, the lookup returns \textit{false}.
A $\del$ uses $\func{try\_delete}$, which transforms
$\tup{v,\flagmember}$ into $\tombstone$,
and overwrites a tentative
copy with $\deleted$, to request that their owner clean up.
The function $\func{try\_delete}$ returns \textit{true} if it successfully deletes the key, \textit{false} if it observes that another process deleted a finalized key (in which case the $\del$ operation can also return \textit{false}), and $\bot$ if a copy of $v$ is no longer present.
If neither the forward nor the backward scan deleted anything, the $\del$ operation returns
\textit{false}.

\begin{algorithm}[tb]
{\small
        \DontPrintSemicolon
         \Function{$\func{del\_copy}(v, j)$: \tcp*[f]{Deletes own copy and cleans the cell}}{
         \Repeat{$\Modify(\ds{table}[j],\var{val},\tombstone)$}{
            $\var{val} \gets \Read(\ds{table}[j])$ \;
            \If{$\var{val} = \tup{v,\flagrestart}$}{
            \lIf{$\Modify(\ds{table}[j],\var{val},\tup{v,\flaginsert})$}{
                \Return $\bot$
            }
             $\var{val} \gets \Read(\ds{table}[j])$ \;
         }
         }
        \lIf{$\var{val} = \deleted$}{
            \Return \textit{true}
        }
        \Return \textit{false} \;
        }
        \BlankLine
        \Function{$\func{try\_delete}(v,\var{val},i)$: \tcp*[f]{If the cell contains the key try to delete it}}{
        \While{$\func{val}(\var{val}) = v$}{
        \If{$\var{val} = \tup{v,\flagmember}$}{
            \Return $\Modify(\ds{table}[i], \var{val}, \tombstone)$
        }
        \Else{
            \lIf{$\Modify(\ds{table}[i], \var{val}, \deleted)$}{
            \Return \textit{true} 
            }
        }
        $\var{val} \gets \Read(\ds{table}[i])$
        }
        \Return $\bot$ \;
        }
        \BlankLine
        \Function{$\func{validate\_copy}(v,\var{val},i)$: \tcp*[f]{Detects $v$ and requests revalidation if needed}}{
        \lIf{$\var{val} \in \set{\tup{v,\flagmember},\tup{v,\flagrestart}}$}{
            \Return \textit{true}
        }
        \lIf{$\Modify(\ds{table}[i], \var{val}, \tup{v,\flagrestart})$}{
        \Return \textit{true}
            }
        $\var{val} \gets \ds{table}[i]$ \;
        \lIf{$\func{val}(\var{val}) = v$}{
            \Return \textit{true}
        }
        \Return \textit{false}
        }
    \caption{Helper functions}
    \label{alg:helper}
}
\end{algorithm}

\begin{algorithm}[tb]
{\small
        \DontPrintSemicolon
        \SetNoFillComment
        \Function{$\Read(\var{loc})$:}{
            \Return $\LL(\var{loc})$
        }
        \BlankLine
        \Function{$\Modify(\var{loc},\var{old}, \var{new})$:}{
            \Return $\SC(\var{loc},\var{new})$
        }
        \BlankLine
        \Function{$\func{del\_other\_copy}(\var{v}, \var{val}, \var{i}, \var{j})$:}{
         $\var{cur\_val} \gets \ds{table}[j]$ \;
                \If{$\var{cur\_val} = \tup{\var{v},\flaginsert}$}{
                    \If{$\neg \SC(\ds{table}[\var{i}], \collided)$}{
                        $\var{val} \gets \ds{table}[\var{i}]$ \;
                        \lIf{$\var{val} = \tup{v,\flagmember}$}{
                            \Return \textit{false}
                        }
                    }
                    \Return \textit{true} \;
                }
                \Return \textit{false} \;
        }
    \caption{$\LL$/$\SC$-specific code}
    \label{alg:ll_sc_code}
}
\end{algorithm}

\paragraph{LL/SC and CAS versions.}
Beyond replacing the keywords $\Read$ and $\Modify$ by $\LL$/$\SC$ or by atomic read and $\CAS$ (respectively) as explained above,
the part of the algorithm that eliminates duplicates is slightly different for the two primitives.
In the $\LL$/$\SC$ version (Algorithm~\ref{alg:ll_sc_code}), 
before writing over the losing copy using an $\SC$, the process checks that its own copy has not been eliminated between the $\LL$ and the $\SC$ to the other copy. This ensures the deleting process's copy was present in the table at the same time as the one eliminated.
If this $\SC$ fails,
the process checks if the other copy has been finalized,
and if so, it gives up on its own copy.
If the other copy has not been finalized,
the process continues as though it were able to write $\collided$,
safe in the knowledge that some other process has handled this copy or deleted it.

It can be verified that our algorithm uses $\LL/\SC$ in a \emph{properly-paired manner}, where an $\SC$ targets the same cell as the preceding $\LL$,
avoiding nested or overlapping patterns that are usually not supported on current architectures.

In the $\CAS$ version (Algorithm~\ref{alg:cas_code}), 
duplicate elimination uses a temporary marked state
$\tup{\tup{v,j},\flagcell}$,
where $j$ is  the index of the cell locked by the operation attempting to delete this copy.
The eliminating process uses this state to claim \emph{provisional ownership} of the remote tentative copy before converting it to $\collided$, emulating $\LL$/$\SC$-style reservation using only $\CAS$
in the specific context of our algorithm
(this is not a general emulation of $\LL$/$\SC$ from $\CAS$).
This adds $O(\log m)$ bits; however, if $\log n < \log m$ we can instead use process identifiers.
Because the process executing the $\insr$ effectively ``locks'' the cell for the duration of the operation, process identifiers and cell indices  can be used interchangeably to declare provisional ownership.

\begin{algorithm}[tb]
{\small
        \DontPrintSemicolon
        \SetNoFillComment
        \Function{$\Read(\var{loc})$:}{
            \Return $* loc$ 
        \BlankLine
        }
        \Function{$\Modify(\var{loc},\var{old}, \var{new})$:}{
            \Return $\CAS(\var{loc},\var{old},\var{new})$
        }
        \BlankLine
        \Function{$\func{del\_other\_copy}(v, \var{val}, i, j)$:}{
        \lIf(\tcp*[f]{Reserved by a different insert}){$\var{val} = \tup{\tup{v,k},\flagcell}$ and $k \neq j$}{
            \Return \textit{true}
        }
        \If(\tcp*[f]{Try to reserve the copy}){$\var{val} \neq \tup{\tup{v,j},\flagcell}$}{
            \If{$\neg \CAS(\ds{table}[i], \var{val}, \tup{\tup{v,j},\flagcell})$}{
                $\var{val} \gets \ds{table}[i]$ \;
            \lIf{$\var{val} = \tup{v,\flagmember}$}{
                \Return \textit{false}
            }
            \Return \textit{true}
            }
        }
                $\var{cur\_val} \gets \ds{table}[j]$ \;
                \If(\tcp*[f]{Try to fully delete other copy}){$\var{cur\_val} = \tup{v,\flaginsert}$}{
                    $\CAS(\ds{table}[i], \tup{\tup{v,j},\flagcell}, \collided)$ \;
                    \Return \textit{true}
                }
                \Return \textit{false}
        }
    \caption{$\CAS$-specific code}
    \label{alg:cas_code}
    }
\end{algorithm}

\subsection{Insert Operations that Return $\abort$}
\label{sec:resize}

An $\insr$ operation that returns $\abort$ (in line~\ref{line:abort} of Algorithm~\ref{alg:insert})
has no effect on the abstract state of the dictionary.
Inserts only abort if they complete
a full scan of the hash table without finding an available cell (either $\emptycell$ or $\tombstone$). 

An $\insr$ that aborts
does not necessarily imply that the table is ``logically full''.
This is because $\insr$ operations of the same keys may temporarily take up space in the table by creating duplicate copies  (which will be eliminated prior to their return), or concurrent $\del$ operations may clear up space during the scan.
However, $\insr$ operations do not abort spuriously when space 
is available: 
if an $\insr(v)$ operation starts running when the table 
has free space, and the table does not already contain key $v$, 
then some available cell must be used up by a key before the $\insr$ completes.
This is formalized in the next proposition:

\begin{restatable}{proposition}{noAbort}
\label{prop:no abort}
If an $\insr(v)$ operation is invoked at a configuration in which 
at least one cell in the table is empty or contains a tombstone, 
and no cell contains the key $v$, 
then some available cell in the table is reclaimed by a key 
before the operation returns.
\end{restatable}

In particular, if the table contains a tombstone or an empty cell, 
then an $\insr$ operation running solo from a quiescent configuration 
will not abort.
Aborts suggest that
the hash table is full or nearly full.
Repeated $\abort$ responses may therefore trigger a resize operation,
which can also be initiated preemptively before the table becomes full.

\section{Correctness Proof}
\label{sec:alg-proof}

In this section, we prove that the hash table is linearizable and lock-free, and analyze its amortized runtime complexity. 
Let $U$ be the size of the key domain, 
$n$ the number of processes, 
and $m$ the size of the hash table.

\MainThm*


\subsection{Basic Properties}

We begin by showing that the linear-probing structure is preserved and that, among concurrent insertions of the same key, at most one copy can be fully inserted.
The proof abuses terminology and says that an operation performs a step 
when the process executing it performs the step.

A cell is \emph{non-empty} if it does not contain  $\emptycell$; 
a cell contains the key $v$ if its value is one of $\tup{v,\flaginsert}$, $\tup{v,\flagmember}$, $\tup{v,\flagrestart}$ or $\tup{\tup{v,\ast},\flagcell}$.

We observe that our algorithm preserves the following key property of sequential linear probing.
This property implies that if a key is in the table, it will be found by scanning the table starting from the key’s initial hash position until either the key is found or an empty cell is encountered. Moreover, two copies of the same key are part of the same run.

\begin{proposition}\label{prop:consecutive}
If a cell $i$ contains the key $v$, 
then all cells $j$, $h(v) \leq j \leq i$, 
are not empty.
\end{proposition}

Our algorithm ensures that, at any point in the execution, at most one insert
operation can become the current \emph{winner}, 
i.e., the one whose copy may be promoted to $\tup{v,\flagmember}$ 
(formalized in Lemma~\ref{lem:onemem} below).
This invariant yields a well-defined ordering of successful insert 
and delete operations, whose linearization points are defined in
Section~\ref{sec:proof-mutator}. 

Consider an $\insr(v)$ operation $op$ that writes  $\tup{v,\flaginsert}$ 
to a cell in the table.
This value can be overwritten with $\tup{v,\flagmember}$, $\deleted$, $\collided$, $\tombstone$ or $\tup{v,\flagrestart}$;
in the $\CAS$ version it can also be overwritten with $\tup{\tup{v,\ast},\flagcell}$.
If it is overwritten with one of the first three values, 
that value can be subsequently overwritten only with $\tombstone$.
If it is overwritten with $\tup{v,\flagrestart}$ 
or $\tup{\tup{v,\ast},\flagcell}$, then $op$ may overwrite the value again with $\tup{v,\flaginsert}$.

Each value in the sequence $\tup{v,\flaginsert}$, $\tup{v,\flagrestart}$ 
and $\tup{\tup{v,\ast},\flagcell}$, 
starting from the first $\tup{v,\flaginsert}$ value written by $op$, 
is a \emph{copy} of $op$.
The sequence cycles through these states with restricted transitions: while $\tup{\tup{v,\ast},\flagcell}$ may transition to either $\tup{v,\flaginsert}$ or $\tup{v,\flagrestart}$, $\tup{v,\flagrestart}$ is restricted to transition only back to $\tup{v,\flaginsert}$.
The cell in which the copies of $op$ are written is called $op$'s cell.
A copy of $op$ is \emph{inserted} if it is overwritten with $\deleted$ or $\tup{v,\flagmember}$.
If a copy is overwritten with $\tup{v,\flagmember}$, 
we call it the \emph{inserted copy} of operation $op$.
A copy of $op$ is \emph{deleted} if it is overwritten with 
$\deleted$, $\collided$ or $\tombstone$,
or if the inserted copy of $op$ is overwritten with $\tombstone$.
Note that only the last copy written by $op$ can be inserted or deleted.

The next lemma formalizes a key invariant mentioned in Section~\ref{sec:alg overview}, namely, that among all simultaneous attempts to insert duplicate copies of $v$, only one can be successful.
Intuitively, this is because before succeeding, an insert scans
the probe sequence and observes that no copy closer to the start of the probe sequence exists. 
In this scan, the insert ignores copies with the flag $\flagrestart$ or $\flagcell$. 
It can do so as the other insert operation must first rewrite the copy with the flag $\flaginsert$ and then scan the table, ensuring it will observe the current operation's copy. 

\begin{restatable}{lemma}{onemem}
\label{lem:onemem}
At any point in the execution, for any key $v$,
at most one cell in the table contains $\tup{v,\flagmember}$.
\end{restatable}

\begin{proof}
Assume, by way of contradiction, that two cells $i$ and $j$, $i \neq j$,
contain  $\tup{v,\flagmember}$ at the same point in $\alpha$.
Let $op_i$ and $op_j$ be the $\insr(v)$ operations that 
write these $\tup{v,\flagmember}$ values to cells $i$ and $j$, respectively.
Assume, without loss of generality, that $op_i$ writes $\tup{v,\flagmember}$ 
to cell $i$ before $op_j$ writes $\tup{v,\flagmember}$ to cell $j$.
    
An $\insr$ writes $\tup{v,\flagmember}$ to a cell only after it writes $\tup{v,\flaginsert}$ to that cell. 
Consider the last time $op_i$ writes $\tup{v,\flaginsert}$ to cell $i$,
and note that it then scans a run starting from $h(v)$ 
before writing $\tup{v,\flagmember}$ to cell $i$.
We next assume the run includes cell $j$ and consider the possible values $op_i$ reads from cell $j$.

First, assume that $op_i$ reads $\tup{v,\flaginsert}$ written by $op_j$ from $j$.
If cell $j$ is closer to the initial hash than cell $i$
(that is, $j -  \func{h}(\var{v}) < i - \func{h}(\var{v})$), 
then $op_i$ deletes its own copy and never writes $\tup{v,\flagmember}$ to cell $i$.
Otherwise, $op_i$ tries to delete the copy of $op_j$. 
Regardless of whether $op_i$ succeeds in deleting the copy, 
the value in cell $j$ changes.
Since $op_i$ did not write $\tup{v,\flagmember}$ to cell $i$ yet and is the first operation to write $\tup{v,\flagmember}$ among the two,
the value of cell $j$ cannot change to $\tup{v,\flagmember}$ and may only change to $\collided$, $\deleted$, $\tup{v,\flagrestart}$, or $\tup{\tup{v,\ast}, \flagcell}$. 
The first two values are not possible, 
since cell $j$ is eventually overwritten with $\tup{v,\flagmember}$.
Hence, either $\tup{v,\flagrestart}$ or $\tup{\tup{v,\ast}, \flagcell}$ 
are written in cell $j$ after $op_i$ reads cell $j$.

Otherwise, $op_i$ reads from cell $j$
either $\tup{v,\flagrestart}$ copy of $op_j$,
or $\tup{\tup{v,\ast},\flagcell}$ copy of $op_j$, 
or the first write of $\tup{v,\flaginsert}$ to cell $j$ by $op_j$
follows the read of cell $j$ by $op_i$.
In all cases, the last write of $\tup{v,\flaginsert}$ to cell $j$ by $op_j$
(which exists since this last copy is overwritten with $\tup{v,\flagmember}$)
follows $op_i$'s read of cell $j$.

If the run excludes cell $j$ and $op_i$ does not read cell $j$, by Proposition~\ref{prop:consecutive}, the first time $op_j$ writes  $\tup{v,\flaginsert}$ to cell $j$ follows the last time $op_i$ writes $\tup{v,\flaginsert}$ to cell $i$.

Regardless of whether $op_i$ reads cell $j$ or not,
$op_j$ reads cell $i$ after the last time $op_i$ writes $\tup{v,\flaginsert}$ to cell $i$.
The last such read of cell $i$ by $op_j$ occurs after $op_j$ writes  $\tup{v,\flaginsert}$ to cell $j$ for the last time. 
Since this is the last such write and $op_j$ writes $\tup{v,\flagmember}$ 
to cell $j$, this value can only change to be $\tup{v,\flagmember}$.
This read returns either $\tup{v,\flaginsert}$ 
or $\tup{v,\flagmember}$ written by $op_i$ to cell $i$.
This is because $op_i$ already performed its last write of $\tup{v,\flaginsert}$ to cell $i$, and this value can only change to $\tup{v,\flagmember}$.
Furthermore, by our initial assumption, once $\tup{v,\flagmember}$ is written to cell $i$, it remains there until the write of $\tup{v,\flagmember}$ to cell $j$.

If $op_j$ reads $\tup{v,\flagmember}$ or if $i -  \func{h}(\var{v}) < j - \func{h}(\var{v})$, 
$op_j$ deletes its own value.
Otherwise ($op_j$ reads $\tup{v,\flaginsert}$), 
$op_j$ tries to delete the copy in cell $i$. 
However, this attempt is unsuccessful since $op_i$ inserts its copy. 
Thus, $op_j$ rereads the cell to obtain $\tup{v,\flagmember}$, 
and deletes its own copy.
Since we assume that the read of cell $i$ occurs after the last write by $op_j$ of $\tup{v,\flaginsert}$ to cell $j$, the value of cell $j$ cannot subsequently change to $\tup{v,\flagrestart}$ or $\tup{\tup{v,\ast},\flagcell}$. 
Consequently, $op_j$ succeeds in deleting its own copy.

In all cases, $op_j$ deletes its own copy and does not write  $\tup{v,\flagmember}$, which is a contradiction.
\end{proof}






\subsection{Linearizability}

Let $\alpha$ be a finite execution.
An $\insr$ operation may return $\abort$ if it completes a full scan of the hash table and does not find an available cell.
We ignore $\insr$ operations that return $\abort$, as they do not affect the sequential specification; thus, they can be linearized at any point, provided they respect the real-time order.

We construct the linearization $\pi_v$ of the remaining 
$\insr$, $\del$ and $\lookup$ operations in $\alpha$.
Let $H(\alpha)|_v$ be the projection of the history $H(\alpha)$ onto key $v$; that is, the subsequence of operations in $H(\alpha)$ involving the input key $v$.
Our proof proceeds by constructing independent linearizations $\pi_v$ of $H(\alpha)|_v$, for each key $v$.
This is sufficient since the locality theorem of linearizability 
allows to obtain the linearization of the whole execution
by merging per-key linearizations while respecting real-time order.

For the rest of the linearizability proof, fix a key $v$.

By Lemma~\ref{lem:onemem},
at most one winner can fully insert its copy at any point in the execution. 
Only after a fully-inserted copy is deleted, a new copy can be crowned the new winner.
This gives us a well-defined order for successful insert and delete operations.
When a delete operation deletes a partially inserted key, 
a successful insert and delete happen simultaneously.
When there is also a fully inserted copy of the key at the same time, 
we first place the delete and then the insert.
This ensures that the ordering ends with a successful insert at such a point.
Otherwise, the insert is placed before the delete.

Ordering unsuccessful insert and delete operations, 
as well as lookup operations, is more challenging.
When a lookup or insert operation finds a partially inserted copy of the key, 
there are three possibilities: 
(a) the copy is eventually fully inserted, 
(b) the copy will be deleted because of a later full insertion, or 
(c) the copy will be deleted because of an earlier full insertion.
If the future successful insertion is not yet linearized, 
there is a pending insert operation, 
which we can treat as successful in the linearization. 
Later, when the actual winner is determined, 
the operation’s response is updated accordingly.

An operation $op$ is a \emph{mutator} if it is 
either an $\insr(v)$ operation whose copy is inserted, 
or a $\del(v)$ operation that deletes a copy of some insert operation.
In the latter case, this means that $op$ overwrites  $\tup{v,\flaginsert}$, $\tup{v,\flagrestart}$ or $\tup{\tup{v,\ast},\flagcell}$ 
with $\deleted$ or  $\tup{v,\flagmember}$ with $\tombstone$.
Let $\mop$ be the set of all mutator operations 
in $\alpha$ with key $v$.

Note that a mutator operation does not necessarily complete in $\alpha$, 
however, if it does, it returns \textit{true}.
If an insert$(v)$ observes a partial copy before a winner exists, 
all later operations should treat $v$ as present. 
In the linearization, we let one ongoing insert serve as the successful insert, 
even if it eventually returns \textit{false}.
This captures the forced-success behavior of unsuccessful inserts described in
Section~\ref{sec:alg overview}.

We begin by ordering the mutator operations 
(Section~\ref{sec:proof-mutator}).
Next (Section~\ref{sec:insert-seq}), we introduce \emph{insert sequences} 
and prove key properties that will be used (in Section~\ref{sec:proof-nonmutator}) 
to place non-mutator operations in between mutator operations.

\subsubsection{Ordering the Mutator Operations}
\label{sec:proof-mutator}

We construct the permutation $\mu_v$ of the mutator operations,
by assigning a \emph{linearization point} to each mutator operation $op$, which is a step in $\alpha$, 
between the invocation and the response of $op$. 
Two mutator operations may be linearized at the same step, in which case they are ordered, 
by marking them with $\set{1,2}$, 
namely, which of them is first and which is second. 

If $iop \in \mop$ is an $\insr(v)$ operation,
then a copy of $iop$ is overwritten with $\tup{v,\flagmember}$ or $\deleted$.
Let $\pt$ be the point in the execution where this overwrite happens.
If the copy of $iop$ is overwritten with $\tup{v,\flagmember}$, 
we define $\lin(iop) = \tup{\pt, 1}$.
Otherwise, a copy of $iop$ is overwritten with $\deleted$.
If at $\pt$ there is a cell in the table containing  $\tup{v,\flagmember}$, we define $\lin(iop) = \tup{\pt, 2}$. 
Otherwise, at $\pt$ no cell in the table contains  $\tup{v,\flagmember}$ and we define $\lin(iop) = \tup{\pt, 1}$.

If $dop \in \mop$ is a $\del(v)$ operation,
then $dop$ overwrites $\tup{v,\flagmember}$ with $\tombstone$, or $\tup{v,\flaginsert}$, $\tup{v,\flagrestart}$ or $\tup{\tup{v,\ast},\flagcell}$ with $\deleted$.
Let $\pt$ be the point in the execution where this overwrite happens.
If $dop$ overwrites $\tup{v,\flagmember}$ with $\tombstone$, we define $\lin(dop) = \tup{\pt, 1}$.
Otherwise, $dop$ overwrites  $\tup{v,\flaginsert}$, $\tup{v,\flagrestart}$, or $\tup{\tup{v,\ast},\flagcell}$ with $\deleted$.
If there is a cell in the table containing  $\tup{v,\flagmember}$ at $\pt$, we define $\lin(dop) = \tup{\pt, 1}$. 
Otherwise, at $\pt$ no cell in the table contains  $\tup{v,\flagmember}$ and we define $\lin(dop) = \tup{\pt, 2}$.

Note that at points where a copy is overwritten with $\deleted$ we linearize both an $\insr(v)$ operation and a $\del(v)$ operation.
If during this overwrite there is a cell in the table containing  $\tup{v,\flagmember}$, 
we linearize the delete first and then the insert.
Otherwise, we linearize them in reverse order,
first the insert then the delete.
This ensures that a key is considered to be in the table if and only if it is stored in a cell with the flag $\flagmember$.
We formalize this invariant below in Lemma~\ref{lem:memlog}.

We say that $\tup{\pt_1, i_1}$ \emph{precedes} $\tup{\pt_2, i_2}$ if either $\pt_1$ precedes $\pt_2$ in $\alpha$, or $\pt_1 = \pt_2$ and $i_1 < i_2$.
Let $\mu'_v$ be a permutation of the operations in $\mop$, 
ordered such that ${op}_1$ precedes ${op}_2$ in $\mu'_v$ 
if and only if $\lin({op}_1)$ precedes $\lin({op}_2)$.
All mutator operations, including those that did not return in $\alpha$, 
are assigned the response \textit{true} in $\mu'_v$.

We next formalize the behavior described in Section~\ref{sec:alg overview}, 
where an insert operation that observes a partial copy may force the key to be treated as present, even if no insert has yet completed.
We construct $\mu_v$ by possibly extending $\mu'_v$.
If at the end of $\alpha$: 
(1) there is a cell in the table containing  $\tup{v,\flaginsert}$, $\tup{v,\flagrestart}$ or $\tup{\tup{v,\ast},\flagcell}$, 
(2) $\mu'_v$ does not end with an $\insr(v)$ operation, and 
(3) there is no mutator $\insr(v)$ operation with a linearization point that occurs at the last step of $\alpha$, we append such an operation as follows.
Note that, by the code, the $\insr(v)$ operation this copy belongs to is pending.
Among all pending $\insr(v)$ operations that write  $\tup{v,\flaginsert}$ in $\alpha$,
let $op$ be the one invoked the earliest.
We obtain $\mu_v$ by appending $op$ to $\mu'_v$ with response \textit{true}. 
Furthermore, we add $op$ to $\mop$, and define its linearization point to be the last step in the execution.
If these conditions are not met, we simply define $\mu_v = \mu'_v$.


Next, we show that, when ignoring the last insert possibly appended to $\mu_v$, $v$ is logically in the table if and only if the table contains  $\tup{v,\flagmember}$.

\begin{restatable}{lemma}{memlog}
\label{lem:memlog}
$\mu'_v$ ends with $\insr(v)$ if and only if 
some cell contains $\tup{v,\flagmember}$ after $\alpha$.
\end{restatable}

\begin{proof}
We prove the lemma by induction on the number of steps in an execution 
prefix $\alpha'$
that 
either overwrite $\tup{v,\flaginsert}$ with $\tup{v,\flagmember}$, 
or overwrite $\tup{v,\flagmember}$ with $\tombstone$, 
or overwrite $\tup{v,\flaginsert}$, $\tup{v,\flagrestart}$ or $\tup{\tup{v,\ast},\flagcell}$ with $\deleted$. 
We note that this suffices since these \emph{overwrite steps} 
are the only steps that append an operation to $\mu'_v$.
The base case, when no such steps appear in $\alpha'$, is trivial.

Assume the lemma holds for an execution prefix $\alpha'$, 
and consider the possible cases.

\emph{If $\tup{v,\flaginsert}$ is overwritten with $\tup{v,\flagmember}$}, 
then $\insr(v)$ is appended to $\mu'_v$.
Since $\tup{v,\flagmember}$ is written, 
there is a cell containing  $\tup{v,\flagmember}$.

\emph{If $\tup{v,\flaginsert}$, $\tup{v,\flagrestart}$ or $\tup{\tup{v,\ast},\flagcell}$ are overwritten with $\deleted$}, 
then an $\insr(v)$ and a $\del(v)$ are appended to $\mu'_v$.
If at the end of $\alpha'$ there is a cell containing $\tup{v,\flagmember}$, 
$\del(v)$ is appended first followed by $\insr(v)$.
        Since the $\tup{v,\flagmember}$ value is not overwritten,
        until the next overwrite, a cell containing  $\tup{v,\flagmember}$ remains.
        If at the end of $\alpha'$ there is no cell containing $\tup{v,\flagmember}$, 
        $\insr(v)$ is appended first and then $\del(v)$.
        Since no $\tup{v,\flagmember}$ value is written,
        until the next overwrite, no cell contains  $\tup{v,\flagmember}$.

\emph{If $\tup{v,\flagmember}$ is overwritten with $\tombstone$}, 
then $\del(v)$ is appended to $\mu'_v$.
Since a $\tup{v,\flagmember}$ value is deleted, 
and by Lemma~\ref{lem:onemem},
no other cell also contains  $\tup{v,\flagmember}$,
no cell contains $\tup{v,\flagmember}$.
\end{proof}



\begin{restatable}{lemma}{lemnmop}
\label{lemma:spec nmop}
$\mu_v$ respects the sequential specification of a dictionary.
\end{restatable}
\begin{proof}
Consider first $\mu'_v$ and assume by way of contradiction, 
that it does not respect the sequential specification of a dictionary. 
Let $op$ be the first operation in $\mu'_v$ 
that violates the sequential specification. 
If $op$ is an $\insr(v)$, then it appears after another $\insr(v)$.
If $op$ is a $\del(v)$, then it is not preceded by an $\insr(v)$.
Consider the overwrite in $\alpha$ that determines the placement of $op$.

\emph{If a $\tup{v,\flaginsert}$ is overwritten by $\tup{v,\flagmember}$}, 
then $op$ is an $\insr(v)$.
By Lemma~\ref{lem:onemem}, 
no cell contains $\tup{v,\flagmember}$ before this overwrite,
and by Lemma~\ref{lem:memlog}, $\mu'_v$ constructed up until before 
this overwrite, does not end with an $\insr(v)$, 
which is a contradiction.

\emph{If a $\tup{v,\flagmember}$ is overwritten with $\tombstone$}, 
then $op$ is a $\del(v)$.
By Lemma~\ref{lem:memlog}, $\mu'_v$ constructed up until before this overwrite, ends with an $\insr(v)$, 
which is a contradiction.

Finally, assume \emph{$\tup{v,\flaginsert}$, $\tup{v,\flagrestart}$ 
or $\tup{\tup{v,\ast},\flagcell}$ are overwritten by $\deleted$}.
If there is a cell containing $\tup{v,\flagmember}$ before this overwrite,
then by Lemma~\ref{lem:memlog}, 
$\mu'_v$ constructed up until before this overwrite, 
ends with an $\insr(v)$.
Thus, appending $\del(v)$ and then $\insr(v)$ 
does not violate the sequential specification. 
Otherwise, no cell contains  $\tup{v,\flagmember}$ before this overwrite,
and by Lemma~\ref{lem:memlog}, 
$\mu'_v$ constructed up until before this overwrite, 
does not end with an $\insr(v)$.
Thus, appending $\insr(v)$ and then $\del(v)$ 
does not violate the sequential specification. 

Therefore, $\mu'_v$ respects the sequential specification.

Finally, assume that an $\insr(v)$ is appended to $\mu'_v$ 
to obtain $\mu_v$.
This happens only if 
$\mu'_v$ does not end with an $\insr(v)$, and therefore, 
appending the $\insr(v)$ obeys the sequential specification.
\end{proof}

\subsubsection{Insert Sequences}
\label{sec:insert-seq}

A copy of an insert operation is either still present in the table, inserted, or is deleted because of a different copy of $v$.
This yields a sequence of insert operations where a copy of each operation is deleted due to the next one in the sequence; 
the sequence ends when a copy of the final insert operation is successfully inserted.

A copy of an $\insr(v)$ operation $op_1$ is deleted because of an $\insr(v)$ operation $op_2$, denoted $op_2 \delrel op_1$, 
if either (1) $op_2$ overwrites a copy of $op_1$ with $\collided$, or
(2) $op_1$ overwrites its own copy with $\tombstone$ after finding a (possibly inserted) copy of $op_2$.
By the code, $op_2$ writes its first copy before $op_1$'s copy is deleted.

Lemma~\ref{lem:not-stale-del} shows a stronger property:
a copy cannot be deleted by an insert operation whose copy has already been deleted.
The details of the next proof are the main place where 
the correctness proof differs between the LL/SC and the
CAS versions of the algorithm. 

\begin{restatable}{lemma}{notstaledel}
\label{lem:not-stale-del}
    Let $op_1$ and $op_2$ be two $\insr(v)$ operations such that $op_2 \delrel op_1$. Then either a copy of $op_2$ and the last copy of $op_1$ are both present in the table at the same time, or the last copy of $op_1$ is a $\tup{\tup{v,\ast},\flagcell}$ value and a copy of $op_2$ and the second-to-last copy of $op_1$ are both present in the table at the same time.
\end{restatable}

\begin{proof}[Proof for the $\LL/\SC$ version]
    If $op_1$ deletes its own copy, it does so after reading a copy of $op_2$.
    Consider the last such read by $op_1$, then $op_1$'s next step is to overwrite its copy with $\tombstone$. 
    Assume that the copy of $op_1$ that is in the table during this last read is overwritten before a copy is deleted.
    Then, as $op_1$'s copy is not deleted yet, and this overwrite is not performed by $op_1$,
    it can only be overwritten with $\tup{v,\flagrestart}$.
    This copy cannot be overwritten until $op_1$'s next step.
    However, $\tup{v,\flagrestart}$ cannot be overwritten with $\tombstone$, in contradiction.

    If $op_2$ deletes a copy of $op_1$, then it begins by reading $op_1$'s copy with an $\LL$. 
    Then, it reads its own cell and gets $\tup{v,\flaginsert}$. 
    Lastly, it overwrites $op_1$'s copy with an $\SC$.
    If the $\SC$ succeeds, then it overwrites the same copy that it initially read with the $\LL$. 
    Therefore, this copy is the last copy of $op_1$.
    In between, it read its own copy, ensuring that a copy of $op_2$ is present in the table at the same time as $op_1$'s last copy.  
\end{proof}

\begin{proof}[Proof for the $\CAS$ version]
    Assume that $op_1$ deletes its own copy.
    Consider the last read of $op_2$'s copy by $op_1$, then $op_1$'s next step is to overwrite its copy with $\tombstone$. 
    If $op_1$'s copy that is in the table during this last read is overwritten, it must be with $\tup{v,\flagrestart}$ or $\tup{\tup{v,\ast},\flagcell}$.
    This is because $op_1$'s copy is not deleted yet, and this overwrite is not performed by $op_1$.
    We already showed for the $\LL$/$\SC$ version that an overwrite with $\tup{v,\flagrestart}$ is not possible, 
    as the overwrite with $\tombstone$ must fail.
    Assume that the copy is overwritten with $\tup{\tup{v,\ast},\flagcell}$.
    Since $op_1$ deletes its own copy, only $op_1$ can overwrite this value.
    Any other overwrite would mean that some other operation deleted this copy, which cannot happen.
    Thus, $op_1$ overwrites this copy with $\tombstone$, 
    and we get that $op_2$'s copy is present in the table at the same time as $op_1$'s second-to-last copy.
    
    If $op_2$ deletes a copy of $op_1$, it first overwrites $op_1$'s copy with $\tup{\tup{v,i},\flagcell}$, where $i$ is $op_2$'s cell.
    Then, it reads $\tup{v,\flaginsert}$ from its own cell. 
    Finally, it overwrites the value $\tup{\tup{v,i},\flagcell}$ in $op_1$'s cell with $\collided$.
    At any point during the execution, only a single insert operation owns a given cell.
    This means that only $op_1$ and $op_2$ can overwrite the $\tup{\tup{v,i},\flagcell}$ copy of $op_1$.
    Since $op_2$ is the only operation that can write the cell number $i$ to the table, 
    the value $\tup{\tup{v,i},\flagcell}$ written by $op_2$ is the same one overwritten with $\collided$, and this is the last copy of operation $op_1$.
    In between, $op_2$ read its own copy, ensuring that a copy of $op_2$ is present in the table at the same time as $op_1$'s last copy.  
\end{proof}

An \emph{insert sequence} $op_k \delrel \ldots \delrel op_1$, $k\geq 1$, \emph{has ended} if a copy of operation $op_k$ is inserted.
The insert sequence is \emph{maximal} if it has ended or 
if a copy of operation $op_k$ is still present in the table at the end of $\alpha$, that is, its final copy is neither deleted nor inserted.
This is denoted by $op_k \delrelstar op_1$.

Let $c_j$ be the cell of $op_j$, $1 \leq j \leq k$.
We next show that a copy is deleted only if its cell is further from the start of the probe sequence than the cell of the deleting operation (with the possible exception of the final deletion).

\begin{lemma}
\label{lem:insertseq-earlier}
Consider a maximal insert sequence $op_k \delrel \ldots \delrel op_1$, $k\geq 1$.
If $1 \leq j < k-1$, then 
$\parens*{c_{j+1} -  \func{h}(\var{v})} < \parens*{c_j - \func{h}(\var{v})}$.
Furthermore, if the sequence has not ended, or if a copy of operation $op_k$ is overwritten with $\deleted$, then 
$\parens*{c_{k} -  \func{h}(\var{v})} < \parens*{c_{k-1} - \func{h}(\var{v})}$.
\end{lemma}

\begin{proof}
Note that no copy of $op_{j+1}$ is overwritten with $\tup{v,\flagmember}$.
By the code and since $op_{j+1} \delrel op_j$, 
a copy of $op_j$ is deleted after $op_j$ finds a copy closer to the start of the probe sequence, or $op_{j+1}$ deletes a copy further from the start.
Therefore, 
$\parens*{c_{j+1} -  \func{h}(\var{v})} < \parens*{c_j - \func{h} (\var{v})}$,
as needed. 
If the sequence has not ended, or if the copy of $op_k$ is overwritten with $\deleted$, a copy of $op_{k-1}$ cannot be deleted after $op_{k-1}$ finds an inserted copy of $op_{k}$.
As before, $c_{k}$ is closer to the start of $v$'s probe sequence than $c_{k-1}$.
\end{proof}

The next lemma is a key property of an insert sequence. 

\begin{restatable}{lemma}{insertseq}
\label{lem:insertseq}
Consider a maximal insert sequence $op_k \delrel \ldots \delrel op_1$, $k > 1$.
If a copy of operation $op_k$ is not inserted, 
then among all operations $op_j$ that have already written a copy to the table, no copy of the operation with the minimal distance 
$\parens*{c_j - \func{h}(\var{v})}$ is deleted.
\end{restatable}

\begin{proof}
The proof is by induction on the number of writes
of the first copies of $op_1, \ldots op_k$, 
before the copy of $op_k$ is inserted.
For the base case, note that the first $\tup{v,\flaginsert}$ copy 
is deleted only because of a copy that is written after it.
Thus, until the next write, the cell contains either $\tup{v,\flaginsert}$, $\tup{v,\flagrestart}$ or $\tup{\tup{v,\ast}, \flagcell}$, or a copy of the operation is inserted, 
implying that $op_k$'s copy is inserted.

For the induction step, assume that the lemma holds until the 
$j$-th write and that the copy of $op_k$ is not overwritten 
with $\tup{v,\flagmember}$ or $\deleted$.
Consider the next, $j$th write of a first copy. 
By Lemma~\ref{lem:insertseq-earlier}, 
a copy is deleted due to an $\insr(v)$ operation 
only if it is further from the start of the probe sequence than the copy of the deleting operation. 
Therefore, the copy residing in the cell closest to the start is deleted only because of a copy that is written after it.
In the next write, the cell contains either $\tup{v,\flaginsert}$, $\tup{v,\flagrestart}$ or $\tup{\tup{v,\ast}, \flagcell}$, or a copy of the operation is inserted, implying that $op_k$'s copy is inserted.
\end{proof}


\subsubsection{Placing the Non-Mutator Operations}
\label{sec:proof-nonmutator}

We can now construct the permutation $\pi_v$ by  
going through the rest of the completed operations,
in the order in which they are invoked in the execution,
and inserting them between the mutator operations,
already placed in $\mu_v$.
We place each such operation $op$ \emph{immediately} before or after a mutator operation $op'$; 
in the latter case, we place $op$ immediately before the next mutator operation that follows $op'$ in $\mu_v$, 
or at the end of the permutation if no such operation exists.
This strategy ensures that operations placed after the same mutator operation respect the real-time order among them.

Let $\nmop$ denote the set of the non-mutator operations that complete in $\alpha$,
and it has two subsets:
$\nmop^+$ are the operations that \emph{find the key}, that is, $\insr(v)$ operations that return \textit{false} and $\lookup(v)$ operations that return \textit{true},
while $\nmop^-$ are the operations that \emph{do not find the key}, that is, $\del(v)$ and $\lookup(v)$ operations that return \textit{false}.

We place these non-mutator operations by examining where the mutator operations are linearized. Specifically, for an operation in $\nmop^+$ (respectively, $\nmop^-$), we search for the earliest valid placement between the last mutator linearized before the operation begins and the first mutator linearized after it ends, such that $v$ is present (respectively, absent) according to the sequential specification.


We start with the operations in $\nmop^+$.
If there is a cell containing  $\tup{v,\flagmember}$ 
when an operation $op$ in $\nmop^+$ is invoked, 
we place $op$ in $\pi_v$ after the mutator operation with the latest linearization point occurring before $op$'s invocation.
If there is a mutator $\insr(v)$ operation with a linearization point occurring between $op$'s invocation and response, 
we place $op$ after the earliest such point. 
If none of these two cases hold, we prove that there is 
an $\insr(v)$ in $\mop$ whose linearization point occurs after $ op$'s response.
Among these operations, let $iop$ be the operation whose linearization point is the earliest;
we place $op$ in $\pi_v$ after $iop$.

Lemma~\ref{lem:nop+} 
shows that whenever a non-mutator
insert or lookup operation returns after observing a key and attempting to replace it with a $\flagrestart$ indication
and the first two cases do not hold,
there is an insert sequence whose final operation is ordered after it.
This is used to show that the ordering of operations in $\nmop^+$ 
is well defined.
The lemma also provides additional properties  
that are later used to prove that the ordering of these operations respects the real-time order.

\begin{restatable}{lemma}{noplus}
\label{lem:nop+}
    Consider $op \in \nmop^+$,
    and assume no cell contains  $\tup{v,\flagmember}$ when $op$ is invoked, and there is no mutator $\insr(v)$ operation with a linearization point between the invocation and the response of $op$.
    Then,
    (a) there is a mutator $\insr(v)$ operation with a linearization point occurring after $op$ returns.\\
    Furthermore, let $iop \in \mop$ be the $\insr(v)$ operation with the earliest linearization point occurring after $op$'s response. 
    Then,
        (b) $iop$ does not follow $op$, and 
        (c) a cell containing key $v$ is present in the table throughout the interval between $op$'s response and $\lin(iop)$.
\end{restatable}

\begin{proof}
We pick an $\insr(v)$ operation $wop$ that writes $\tup{v,\flaginsert}$ to the table, so that $op$ is pending while a copy of $wop$ is present in the table.
    
    If $op$ is a $\lookup(v)$ or an $\insr(v)$ where no write to the table occurs, then $op$ returns after finding a cell containing the key $v$ and possibly trying to overwrite it with $\tup{v,\flagrestart}$.
    If $op$ finds $\tup{v,\flagmember}$ or $\tup{v,\flagrestart}$, let this value be a (possibly inserted) copy of operation $wop$.
    If $op$ finds $\tup{v,\flaginsert}$ or $\tup{\tup{v,\ast},\flagcell}$ and successfully overwrites it with $\tup{v,\flagrestart}$, let both of these values be copies of operation $wop$.
    If the overwrite fails, it rereads the cell for the last time and finds a value with key $v$. 
    Let this value be the (possibly inserted) copy of operation $wop$.
    
    Otherwise, if $op$ is an $\insr(v)$ that returns \textit{false} after writing $\tup{v,\flaginsert}$, then let $wop$ be operation $op$ itself.
    Let $op_k$ be the final operation in the maximal insert sequence starting with $wop$ (i.e., $op_k \delrelstar wop$). 
    We start by proving (a).

    If no copy of $op_k$ is ever inserted, a copy of $v$ is present in the table at the end of $\alpha$.
    By construction, either $\mu_v$ ends with an $\insr(v)$ operation, and so does $\pi_v$, or an $\insr(v)$ operation with a linearization point occurring at $\alpha$'s last step is placed at the end of $\pi_v$.
    Thus, there is a mutator $\insr(v)$ operation whose linearization point is after $op$ returns.
    
    Next, we assume that a copy of $op_k$ is inserted.
    Assume, towards a contradiction,
    that the linearization point of $op_k$ occurs before $op$ returns.
    Since no $\insr(v)$ operation has a
    linearization point occurring between the invocation and the response of $op$, the linearization point by $op_k$ occurs before $op$'s invocation.
    In other words, a copy of $op_k$ is inserted before $op$'s invocation.
    Since no cell contains $\tup{v,\flagmember}$ when $op$ is invoked, $op_k$'s copy is deleted before the invocation of 
    $op$.

    We show that $op$ reads the cell of $op_k$ in its forward scan
    before the write of the penultimate copy of $op_k$.
    However, since no new copy is written after a copy in inserted, 
    this contradicts that a copy of $op_k$ is inserted before the invocation of $op$.

    First, we show that the read of $op$ from the cell of $iop$, 
    in its forward scan, precedes 
    the write of the penultimate copy of $iop$
    (which is its only copy if $iop$ has only one copy).
    If $op = iop$, this holds since $op$ first completes the forward scan and then writes a copy for the first time.
    If $op \neq iop$ and it reads $\tup{v,\flagrestart}$ or successfully overwrites a copy with $\tup{v,\flagrestart}$, then the $\tup{v,\flagrestart}$ is overwritten with $\tup{v,\flaginsert}$ before $iop$'s copy is either inserted or deleted.
    
    If the overwrite of $op$ with $\tup{v,\flagrestart}$ fails, it rereads the cell and finds a value containing the key $v$.
    This value cannot be $\tup{v,\flagmember}$.
    This is because no cell contains $\tup{v,\flagmember}$ when $op$ is invoked and no insertion of $v$ occurs before $op$ returns.
    So, this value is either $\tup{v,\flaginsert}$, $\tup{v,\flagrestart}$ or $\tup{\tup{v,\ast},\flagcell}$. 
    Then, either another copy of the same operation is written to the cell,
    or this copy is deleted and only then the first copy of $iop$ is written to this cell.

    Assume that for some $j$, $2\leq j \leq k$, 
    the write of the penultimate copy of $op_j$
    precedes the read of the cell of $op_j$ in the forward scan of $op$; assume $j$ is the first such operation in the insert sequence. 
    Since $op_j \delrel op_{j-1}$, 
    by Lemma~\ref{lem:not-stale-del}, a copy of $op_j$ 
    is in the table after the write of the penultimate copy of  $op_{j-1}$.
    Thus, no copy of $op_j$ is deleted 
    right after $op$ reads the cell of $op_{j-1}$.
    
    If $j = k$, this immediately leads to a contradiction.
    Otherwise, by Lemma~\ref{lem:insertseq-earlier}, 
    $c_{j}$ is closer to the start of the probe sequence than $c_{j-1}$.
    Thus, in its forward scan, $op$ first reads $c_j$ and then $c_{j-1}$.
    So, the write of the penultimate copy of $op_j$ 
    precedes the read of $c_j$ in the forward scan of $op$, 
    and no copy of $op_j$ is deleted before this read.
    Then, $op$ encounters a copy of $op_j$ in its forward scan.
    Since $j \neq k$, no copy of $op_j$ is inserted.
    Consequently, since $op$ attempts to overwrite this copy with $\tup{v,\flagrestart}$, it follows that either two additional copies of $op_j$ are written to the cell, or a copy of $op_j$ is deleted before $op$ continues to the next cell. Both scenarios yield a contradiction.

    To prove (b), assume towards a contradiction that $iop$ follows $op$.
    Therefore, since a copy of $wop$ is written before $op$ returns, $iop$ is invoked after a copy of $wop$ is first written.
    By (a), a copy of ${op}_k$ is inserted after $op$'s response, or not inserted at all. 
    Thus, either $iop = {op}_k$ or a copy of $iop$ is inserted before a copy of ${op}_k$ is inserted.
    In both cases, the backward scan of $iop$ starts after a copy of $wop$ is written and ends before a copy of ${op}_k$ is inserted.
    
    Since ${op}_k \delrelstar wop$, 
    Lemma~\ref{lem:insertseq} implies that after a copy of operation $wop$ is first written and until a copy of operation ${op}_k$ is inserted, the key $v$ always appears in the table and can move only to a cell closer to the start of the probe sequence.
    This implies that $iop$ finds the key $v$ in its initial backward scan, which contradicts the fact that it is a mutator operation.

    Finally, we prove (c).
    By (a), a copy of ${op}_k$ is inserted after $op$'s response, or not inserted at all. Thus, either $iop = {op}_k$ or a copy of $iop$ is inserted before a copy of ${op}_k$ is inserted.
    Since ${op}_k \delrelstar wop$, 
    by Lemma~\ref{lem:insertseq}, 
    between the point where a copy of operation $wop$ is first written and until a copy of operation ${op}_k$ is inserted, the key $v$ always appears in the table.
    This implies that throughout the interval between $op$'s response and until a copy of operation ${op}_k$ is inserted, there is a cell containing 
    a copy of $v$ belonging to one of the operations in the insert sequence.
    Since the interval between $op$'s response and $\lin(iop)$ is contained in this interval, the claim follows.
\end{proof}

\begin{restatable}{lemma}{nmopplus}
\label{lemma:spec nmop plus}
The linearization of $\nmop^+$ respects the sequential specification of a dictionary.
\end{restatable}

\begin{proof}
We prove that every operation $op$ in $\nmop^+$
is placed immediately after a mutator $\insr(v)$ 
operation in $\pi_v$.

If there is a cell containing $\tup{v,\flagmember}$ 
when $op$ is invoked, then by Lemma~\ref{lem:memlog}, $\mu'_v$ constructed up to this point ends with an $\insr(v)$ operation.
Thus, the last mutator operation placed up until this point in $\mu_v$, and consequently also in $\pi_v$, is $\insr(v)$ and $op$ is placed after it in $\pi_v$.

If there is a mutator $\insr(v)$ operation whose linearization point in between the invocation and the response of $op$, 
then $op$ is placed after a mutator $\insr(v)$ operation.

Finally, if neither of these two cases holds, by Lemma~\ref{lem:nop+}(a), there is a mutator $\insr(v)$ operation with a linearization point occurring after $op$ returns, and $op$ is placed after the one with the earliest linearization point.
\end{proof}

Consider now an operation $op$ in $\nmop^-$.
If no cell contains the key $v$ when $op$ is invoked, 
or no copy of $v$ is deleted between its invocation and response, we place $op$ in $\pi_v$ before the next mutator operation whose linearization point is after the invocation of $op$.
If the first case does not hold, we prove that there is 
a $\del(v)$ in $\mop$ whose linearization point occurs between $op$'s invocation and response.
Hence, if neither of these conditions holds,
let $dop$ in $\mop$ be the $\del(v)$ operation with the earliest linearization point 
between the invocation and response of operation $op$.
We place $op$ after $dop$ in $\pi_v$.

\begin{restatable}{lemma}{keynotfound}
\label{lem:keynotfound}
    Let $op \in \nmop^-$.
    Suppose that when $op$ is invoked, the table contains a (possibly inserted) copy of $v$  belonging to operation $op_1$.
    Let ${op}_k \delrelstar {op}_1$, then a copy of operation ${op}_k$ is inserted and deleted before $op$ returns.
\end{restatable}

\begin{proof}
    First, assume that when $op$ is invoked, the table contains a $\tup{v,\flagmember}$ copy of $op_1$.
    Suppose for the sake of contradiction that this value is not deleted before $op$ returns.
    According to the algorithm, $op$ should encounter this value.
    If $op$ is a $\lookup(v)$ operation, it returns \textit{true}.
    If $op$ is a $\del(v)$ operation, it either successfully overwrites this value with $\tombstone$ and returns \textit{true}, or fails and returns \textit{false}.
    In the latter case, the failure implies that the $\tup{v,\flagmember}$ value is overwritten with $\tombstone$.
    In all cases, this contradicts the assumption that the value is not deleted.

    Next, assume that when $op$ is invoked, the table contains a copy of $op_1$ with value $\tup{v,\flaginsert}$, $\tup{v,\flagrestart}$, or $\tup{\tup{v,\ast},\flagcell}$.
    
    If a copy of operation ${op}_k$ is not deleted before $op$ returns,
    by Lemma~\ref{lem:insertseq}, there is always a copy of key $v$ in the table,
    during both the forward and backward scans of the operation.
    Furthermore, if a copy of $v$ is deleted, there is a copy of $v$ in a cell closer to the start of the probe sequence. 
    Hence, a copy of $v$ is found during the backward scan of $op$.
    
    If $op$ is a $\lookup(v)$ operation, this implies that it returns \textit{true}.

    If $op$ is a $\del(v)$ operation, the copy cannot move closer to the start once it reaches the beginning of the probe sequence. There, $op$ finds it, successfully overwrites it with $\deleted$, and returns \textit{true}.
    However, this contradicts the assumption that $op$ returns \textit{false}.
    Hence, a copy of operation ${op}_k$ is deleted before $op$ returns.

    Note that, by Lemma~\ref{lem:insertseq}, until a copy of ${op}_k$ is deleted, there is always a copy of $v$ in the table that can move only to a cell closer to the start of the probe sequence.
    
    If ${op}_k$'s copy is overwritten
    with $\deleted$, the claim holds immediately.
    Otherwise, if it is overwritten with $\tup{v,\flagmember}$ before $op$ is invoked, then as shown in the first part of the proof, the copy is deleted before $op$ returns.

    Assume ${op}_k$'s copy is overwritten with $\tup{v,\flagmember}$ after $op$'s invocation.
    We consider three cases, and in all of them we show that the $\tup{v,\flagmember}$ value of ${op}_k$ is overwritten with $\tombstone$ before $op$ returns.
    Note that it is enough to show that cell $c_k$ is read in one of the scans of $op$ after ${op}_k$'s copy is first written to the table.
    Since $op$ returns \textit{false}, a copy of ${op}_k$ is either not found or $op$ tries to overwrite it and fails.
    Both cases imply that $\tup{v,\flagmember}$ is overwritten with $\tombstone$.

    \textbf{Case 1:}
     ${c_{k} - \func{h}(\var{v})} < {c_{k-1} - \func{h}(\var{v})}$.
    This implies that $c_k$ is the cell closest to the start of the probe sequence among all cells in the sequence.
    Then, cell $c_k$ is read in the backward scan after ${op}_k$'s copy is first written; otherwise, this contradicts Lemma~\ref{lem:insertseq}.

    \textbf{Case 2:}
    ${c_{k} -  \func{h}(\var{v})} > {c_{1} - \func{h}(\var{v})}$. 
    This implies that $c_k$ is the cell furthest from the start of the probe sequence among all cells in the sequence.
     
     We show that ${op}_k$'s copy is first written before ${op}_1$'s copy is first written.
     Assume towards a contradiction that this does not hold, 
     and consider the last write of $\tup{v,\flaginsert}$ by ${op}_k$.
     Note that by the initial assumption, this value is subsequently overwritten with $\tup{v,\flagmember}$. 
     
     During the scan performed by ${op}_k$ following this write,  
     the operation encounters a copy of $v$ in a cell closer to the start of the probe sequence. Since ${op}_1$'s copy is already present while ${op}_k$'s copy is not inserted yet, this follows from Lemma~\ref{lem:insertseq}.
    However, this implies that ${op}_k$ deletes its own copy, which contradicts the fact that this copy is eventually inserted.

    Recall that ${op}_1$'s copy is already written when $op$ is invoked.
    Since we showed that ${op}_k$'s copy is written before ${op}_1$'s copy, it follows that ${op}_k$'s copy is also present when $op$ is invoked.
    Consequently, ${op}_k$'s copy is written before cell $c_k$ is read during $op$'s forward scan.

    \textbf{Case 3:}
    For some index $1 \leq j < k-1$, ${c_{j+1} - \func{h}(\var{v})} < {c_{k} -  \func{h}(\var{v})} < {c_{j} - \func{h}(\var{v})}$.
    That is, ${op}_{j+1}$'s cell is closer to the start of the probe sequence than ${op}_k$'s cell, 
    and ${op}_{j}$'s cell is further from it.

    We show that 
    ${op}_k$'s copy is first written before any of the copies of operations ${op}_{\ell}$, $j+1 \leq \ell \leq k -1$, are first written.
    If this is not the case, consider the last write of $\tup{v,\flaginsert}$ by ${op}_k$. Note that by the initial assumption, this value is subsequently overwritten with $\tup{v,\flagmember}$.

    During the scan performed by ${op}_k$ following this write,  
    the operation encounters a copy of $v$ in a cell closer to the start of the probe sequence. 
    Since, for some $j+1 \leq \ell \leq k -1$, 
    ${op}_{\ell}$'s copy is already present while ${op}_k$'s copy is not inserted yet, this follows from Lemma~\ref{lem:insertseq}.
    However, this implies that ${op}_k$ deletes its own copy, which contradicts the fact that this copy is eventually inserted.

    Note that in a backward scan, first cell $c_j$ is read, then cell $c_k$, and finally cell $c_{j+1}$.
    If in $op$'s backward scan it does not return after reading cell $c_{j}$, 
    by Lemma~\ref{lem:insertseq}, either ${op}_k$'s copy is inserted, or a closer cell $c_\ell$, for $j +1 \leq \ell \leq k$, contains a copy of $v$.
    This follows from ${op_1}$'s copy already being present in the table before this scan begins.

    As ${op}_k$'s copy is written before any of the copies of operations ${op}_\ell$, $j+1\leq \ell\leq k-1$, are first written, a copy of ${op}_k$ is written in the cell $c_k$ when $c_k$ is read in $op$'s backward scan.
\end{proof}

\begin{lemma}\label{lemma:spec nmop minus}
The linearization of $\nmop^-$ respects the sequential specification of a dictionary.
\end{lemma}

\begin{proof}
    Let $op \in \nmop^-$.
    If $op$ is placed directly after a $\del(v)$ operation, the lemma clearly holds.
    If no cell contains  $\tup{v,\flagmember}$ at $op$'s invocation point, then by Lemma~\ref{lem:memlog}, $\mu'_v$ constructed up to this point does not end with an $\insr(v)$ operation.
    Thus, the next mutator operation placed after this point in $\mu_v$, and consequently also in $\pi_v$, if exists, 
    should be an $\insr(v)$ operation and $op$ can be placed before it in $\pi_v$.
    If there is a cell containing  $\tup{v,\flagmember}$ at $op$'s invocation point, by Lemma~\ref{lem:keynotfound}, this $\tup{v,\flagmember}$ value is deleted before $op$ returns.
    Thus, by construction, $op$ is placed after a $\del(v)$ operation.
\end{proof}

The proof of the next lemma shows that $\pi_v$ preserves  real-time order.
Together with Lemmas~\ref{lemma:spec nmop}, \ref{lemma:spec nmop plus} and~\ref{lemma:spec nmop minus}, 
it follows that $\pi_v$ is a linearization of $\alpha$ restricted to operations with input $v$.

\begin{restatable}{lemma}{linwhole}
\label{lemma:real-time whole}
For any execution $\alpha$,
    $\pi_v$ is a valid linearization of
    $H(\alpha) |_v$.
\end{restatable}

\begin{proof}
Lemmas~\ref{lemma:spec nmop}, \ref{lemma:spec nmop plus}, and~\ref{lemma:spec nmop minus} imply that $\pi_v$ respects the sequential specification. 
It remains to prove that $\pi_v$ preserves the real-time order.

The real-time order of operations in $\mop$ is preserved in $\pi_v$,
since mutator operations are placed in $\pi_v$ at a point between their invocation and response.
If a final $\insr(v)$ is appended to $\mu'_v$, 
resulting in $\mu_v$, this operation is pending.
Thus, no operation in $\alpha$ follows it, 
so appending it to the end of $\mu'_v$ 
preserves the real-time order. 

Next, we show that operations in $\nmop$ are placed after a mutator operation whose linearization point is in the interval between
the last linearization point before the operation's invocation and the earliest linearization point after it returns.
For operations in $\nmop^-$, this holds by definition.

Let $op \in \nmop^+$. If there is a cell containing  $\tup{v,\flagmember}$ when $op$ is invoked, or there is a mutator $\insr(v)$ operation with a linearization point occurring between $op$'s invocation and response, the claim holds by definition.
If neither of the cases holds, $op$ is placed after $iop\in mop$, the $\insr(v)$ operation with the earliest linearization point occurring after the response of $op$.
We prove that no other mutator operation is linearized
between the response of $op$ and $\lin(iop)$.

By the definition of $iop$, no mutator $\insr(v)$ operation other than $iop$ is linearized within this interval.
Consider a mutator $\del(v)$ operation $dop$ with a linearization point in this interval.
Recall our assumption that at $op$'s invocation, no cell contains  $\tup{v, \flagmember}$, and no $\insr(v)$ is linearized between $op$'s invocation and response. 
Consequently, no cell contains $\tup{v, \flagmember}$ during the interval between $op$'s response and $\lin(iop)$.
Hence, there is an $\insr$ operation linearized at the same execution step as $dop$ that precedes $dop$ in the linearization order, and 
this $\insr(v)$ operation must be $iop$.
Therefore, the linearization point of $dop$ must fall outside this interval.

Since the placement of operations in $\nmop^-$ is confined to the interval bounded by the latest linearization point prior to the operation's invocation and the operation's response, the real-time order is preserved between operations in $\mop$ and $\nmop^-$, as well as among operations in $\nmop^-$.

Operations in $\nmop^+$ are placed in an interval that includes
the earliest linearization point following the operation's response.
This preserves the real-time order between any two operations in $\nmop^+$.
We next show this also holds for an operation in $\nmop^+$
and operations in $\mop$ or $\nmop^-$.

Let $op^+ \in \nmop^+$. If there is a cell containing  $\tup{v,\flagmember}$ when $op$ is invoked, or there is a mutator $\insr(v)$ operation with a linearization point occurring between $op$'s invocation and response, the placement interval ends at $op$'s response.
Consequently, in this case, the real-time order is naturally preserved between $op^+$ and any operation in $\mop$ or $\nmop^+$.

We next consider the case where neither condition holds. 
Here, $op^+$ is placed after $iop \in \mop$, defined as the $\insr(v)$ operation with the earliest linearization point following the response of $op^+$.
Since we showed that no other mutator operation is linearized between $op^+$'s response and $\lin(iop)$, the only real-time order violation can be between $op^+$ and $iop$.
Yet, Lemma~\ref{lem:nop+}(b) guarantees that $iop$ does not follow $op^+$.

Let $op^-\in \nmop^-$. If $op^+$ follows $op^-$, then the placement intervals do not intersect, and the real-time order is preserved.

Assume $op^+$ precedes $op^-$, but is placed after it.
If $\lin(iop)$ occurs before the invocation of ${op}^-$, ${op}^-$ is placed after ${op}^+$. 
Since ${op}^+$ is placed after ${op}^-$, 
$\lin(iop)$ occurs after the invocation of ${op}^-$.

The invocation of $op^-$ falls within the interval between the response of $op^+$ and $\lin(iop)$. Consequently, by Lemma~\ref{lem:nop+}(c), at the time of this invocation, the table contains a copy of $v$.

By Lemma~\ref{lem:keynotfound}, a copy of $v$ is deleted between ${op}^-$'s invocation and response.
Let $dop \in \mop$ be the $\del(v)$ operation with the earliest linearization point that occurs between ${op}^-$'s invocation and response. ${op}^-$ is placed after $dop$.
We already showed that no other mutator operation is linearized between the response of $op^+$ and $\lin(iop)$.
Hence,  $\lin(iop)$ precedes $\lin(dop)$ in $\alpha$, and $op^+$ is placed before $op^-$ is $\pi_v$.
\end{proof}

    \subsection{Liveness}
\label{app:liveness}

It is easy to see that $\lookup$ operations are wait-free.
For $\del(v)$ and $\insr(v)$
operations, the only cases in which these operations
take infinitely many steps without returning
are:
\begin{itemize}
    \item For $\del(v)$: the operation repeatedly attempts and fails to overwrite a tentative copy of $v$ at a specific cell with $\deleted$;
    \item For $\insr(v)$: the operation
    writes a tentative copy into some cell,
but is then forced to revalidate it infinitely often
by signals from other processes.
\end{itemize}
In both cases,
the operation attempts infinitely many times to modify 
a specific cell containing a copy of $v$.
Thus, to prove lock-freedom, we show that 
this scenario arises only if infinitely many other operations complete.

Since $\lookup$ and $\del$ operations modify a memory cell only once during their execution, their interference is naturally bounded. 
The difficult case is $\insr$ operations, for which we prove two lemmas: 
one bounds the interference from copies of other operations, 
and the other shows that an operation either rewrites its copy 
a bounded number of times or infinitely many operations complete.


\begin{restatable}{lemma}{interdiff}
\label{lemma:bounded interference diff copy}
An $\insr(v)$ operation $op$ performs at most a bounded number of modifying instructions on cells containing a copy of $v$ belonging to other operations.
\end{restatable}

\begin{proof}
From the code, 
$op$ modifies a cell containing a tentative copy \emph{of a different operation} in the following situations:
(1) $op$ overwrites a tentative copy with $\langle v,\flagrestart\rangle$, or
(2) in the $\CAS$ version, $op$ overwrites a tentative copy with $\langle\langle v,j\rangle,\flagcell\rangle$, or
(3) in both the $\CAS$ and $\LL/\SC$ versions, $op$ overwrites a tentative copy with $\collided$.

Case (1) 
can happen only once before $op$ returns.

We show that $op$ performs only a constant number of modifying instructions corresponding to cases (2) and (3), 
before either returning or
moving to a state from which it returns with a finite number of steps.
Assume towards contradiction this is not the case. 
Then, $op$ overwrites a tentative copy with $\collided$ infinitely many times, or, in the $\CAS$ version, with $\langle\langle v,j\rangle,\flagcell\rangle$.

Assume $op$ overwrites a tentative copy with $\collided$ infinitely many times.
Once a tentative copy is overwritten with $\collided$, 
a new tentative copy of the same operation is never written again. 
Thus, $op$ is overwriting tentative copies of distinct operations.

Consider the sequence of operations whose copies are overwritten by $op$. Since this sequence is infinite, we can identify an operation $op'$ in this sequence that is invoked after $op$ has already written its own initial tentative copy.
Before $op$ overwrites a tentative copy with $\collided$, it verifies that its own tentative copy has not been deleted; therefore, a tentative copy of $op$ remains present in the table throughout the execution interval of $op'$.

Before writing a tentative copy, $op'$ first scans a run starting from $h(v)$. By Proposition~\ref{prop:consecutive}, 
during this scan, $op'$ encounters a tentative copy of $op$.
Since no copy of $op$ is ever deleted, $op'$ should return \textit{false} immediately, without writing a new tentative copy to the table. This contradicts the assumption that $op'$ writes a tentative copy, which $op$ later overwrites.

Next, consider the case where $op$ overwrites a tentative copy with $\langle\langle v,j\rangle,\flagcell\rangle$ infinitely many times.
If $op$ overwrites tentative copies belonging to infinitely many distinct operations, the proof proceeds as in the previous case.
Otherwise, $op$ repeatedly overwrites a tentative copy of the same operation $op'$.
This implies that $op'$ repeatedly rewrites its cell to $\langle v,\flaginsert\rangle$ (potentially after being reset to $\langle v,\flagrestart\rangle$).

Since $op$ tries to delete a copy of $op'$,
$op$'s tentative copy is closer to the start of the probe sequence of $v$ than the tentative copy of $op'$.
Furthermore, because $op$ verifies the existence of its own copy before performing the overwrite, the tentative copy of $op$ remains present in the table throughout the execution interval of $op'$.
Each time $op'$ rewrites the value $\tup{v,\flaginsert}$, 
it scans a run again starting from $h(v)$. 
By Proposition~\ref{prop:consecutive}, $op'$ encounters the tentative copy of $op$ during this scan. 
Since $op$'s copy is closer to the start of the probe sequence, $op'$ attempts to delete its own copy. 
This deletion attempt is unsuccessful, as we assumed infinitely many tentative copies of $op'$ are written.
This failure can only occur due to an intervention by a third operation $op''$ overwriting $op'$'s tentative copy with $\langle v,\flagrestart\rangle$.

As this occurs infinitely many times, we can choose an $op''$ invoked after $op$ writes its first tentative copy. 
In $op''$'s forward scan, it encounters $op$'s tentative copy (which is closer to the start of the probe sequence) before reaching $op'$'s tentative copy. Consequently, since no copy of $op$ is ever deleted, $op''$ returns immediately. This contradicts the claim that $op''$ reaches $op'$ tentative copy and overwrites it with $\langle v,\flagrestart\rangle$.
\end{proof}

\begin{restatable}{lemma}{interown}
\label{lemma:bounded interference own copy}
An $\insr(v)$ operation $op$ either performs at most a bounded number of modifying instructions on the cell containing its own copy, or infinitely many other operations return during the execution interval of $op$.
\end{restatable}

\begin{proof}
Consider an $\insr(v)$ operation $op$ on some key $v$.
By the code, 
$op$ modifies a cell containing \emph{its own} tentative copy 
either when (1) $op$ promotes its own copy to $\langle v,\flagmember\rangle$, deletes its own copy by writing $\tombstone$, or
(2) $op$ writes a new copy of $\tup{v,\flaginsert}$.

In case (1), $op$ overwrites its tentative copy with one of these values only once before returning.

In case (2), $op$ infinitely many times overwrites either $\tup{v,\flagrestart}$ with $\tup{v,\flaginsert}$, or, 
in the $\CAS$ version, $\tup{\tup{v,\ast},\flagcell}$ with $\tup{v,\flaginsert}$.

If $op$ overwrites $\tup{v,\flagrestart}$ with $\tup{v,\flaginsert}$ infinitely many times,
then since an $\insr(v)$ or $\lookup(v)$ operation overwrites a tentative copy 
with $\tup{v,\flagrestart}$ at most once before returning,
implying that infinitely many $\insr(v)$ or $\lookup(v)$ 
operations return.

Otherwise, $op$ overwrites $\tup{\tup{v,\ast},\flagcell}$ with $\tup{v,\flaginsert}$ infinitely many times.
By Lemma~\ref{lemma:bounded interference diff copy}, any single operation overwrites a tentative copy with $\tup{\tup{v,\ast},\flagcell}$ only a finite number of times before returning. 
Therefore, an infinite number of such operations return.
\end{proof}


\begin{lemma}
\label{lemma:bounded interference}
For any key $v$, 
an operation on $v$ either performs at most a bounded number of modifying instructions 
to cells containing $v$ before returning, 
or infinitely many operations return during the execution interval of the operation.
\end{lemma}

\begin{proof}
A $\lookup(v)$ overwrites a tentative copy with 
$\langle v,\flagrestart\rangle$ only once before returning.
A $\del(v)$ overwrites a tentative copy with $\deleted$ 
or an inserted copy with $\tombstone$ only once before returning.
For $\insr(v)$, the lemma follows from 
Lemma~\ref{lemma:bounded interference diff copy} 
and Lemma~\ref{lemma:bounded interference own copy}.
\end{proof}

In particular, if an execution has infinitely many such modifying instructions 
on cells containing $v$, 
then infinitely many operations on $v$ return.
We will use this lemma to prove that $\del$ and $\insr$ operations are lock-free.

\begin{lemma}
$\del$ is lock-free. 
\end{lemma}

\begin{proof}
Consider a $\del(v)$ operation $op$ that takes infinitely many
steps without returning.  
Then $op$ repeatedly reads
$\tup{v,\flaginsert}$ or $\tup{v,\flagrestart}$, 
and in the $\CAS$ version also $\tup{\tup{v,\ast},\flagcell}$,
and fails to overwrite it with $\deleted$.
Each such failure indicates that, between $op$'s read 
and its attempted write, 
another operation on $v$ performed a modifying instruction on the cell.
By Lemma~\ref{lemma:bounded interference}, 
each operation on $v$ can perform only a bounded number of 
modifying instructions before returning, or infinitely many other operations return.
\end{proof}

\begin{lemma}
    $\insr$ is lock-free.
\end{lemma}

\begin{proof}
Consider an $\insr(v)$ operation $op$ that takes infinitely many
steps without returning.  

If $op$ never successfully writes $\tup{v,\flaginsert}$, 
it returns after the forward or backward scan, or after traversing all $m$ cells in an attempt to find an available slot.


Otherwise, $op$ writes $\tup{v,\flaginsert}$.
If $op$ takes an infinite number of steps without returning, 
then this value is overwritten infinitely many times 
with $\tup{v,\flagrestart}$, in the $\LL/\SC$ version, 
and with $\tup{v,\flagrestart}$ or $\tup{\tup{v,\ast},\flagcell}$,
in the $\CAS$ version.
Each such overwrite is a modifying instruction to a cell containing~$v$
performed by some operation on~$v$.
By Lemma~\ref{lemma:bounded interference}, 
each operation on~$v$ can perform only a bounded number of such
modifying instructions before returning, 
except when infinitely many other operations return.
Thus, an insert cannot take infinitely many steps 
without infinitely many operations completing.
\end{proof}

    \subsection{Amortized Step Complexity}
\label{sec:time-analysis}


In this section, we analyze the runtime of our algorithm
under the simplifying assumption that at any point in time, there is at most one active $\insr$ operation for any given key. (Multiple concurrent operations on the same key are permitted, as long as only one of them is an insertion.) 
Without this assumption, the adversary can force the creation of tombstones (as demonstrated in Fig.~\ref{fig:alg}), though we note that it can only deterministically force $O(n)$ tombstones in any given run in the hash table.

We recall that the classical analysis of Knuth~\cite{Knuth1963OpenAddressing} for sequential linear-probing hash tables shows that after any sequence of $k$ insertions,
if we represent the resulting load factor on the table by $k/m = 1 - 1/x$ for $x > 0$, then for any key $v$, 
the expected distance from $h(v)$ to the end of $v$'s run is at most $O(x^2)$. 

\paragraph{Mapping to a sequential execution.}
Knuth~\cite{Knuth1963OpenAddressing} did not address tombstones in his analysis (in fact, he addressed only insertions, but his analysis holds if deletions are implemented by reordering the remaining keys to take the empty place of the deleted key, rather than using tombstones).
To address this gap we perform a pessimistic analysis that {ignores} the possibility that an inserted key will take the place of a previously deleted one:
we assume that any inserted remains in the hash table,
ignoring the possible re-use of tombstones that are left behind by successful deletions.

To map concurrent executions to sequential ones, we define the \emph{effective insertion schedule} of an execution $\alpha$. This schedule consists of the subsequence of $\insr$ operations in $\alpha$ that perform a write of $\tup{\ast,\flaginsert}$, ordered by the time of that write.
This implies the following:

\begin{proposition}
\label{prop:maptoseq} 
    Let $\alpha$ be an execution consisting only of $\insr$ operations, where no two concurrent operations have the same key. Then, the state of the table at the end of $\alpha$ is identical to the state after the sequential execution of the effective insertion schedule of $\alpha$.
\end{proposition}

Consider a sequential execution of all the insert operations in the batch. 
By Knuth's analysis the expected run length for any given key is $O(x^2)$.
Crucially, this holds regardless of the execution order.
Thus, the specific effective insertion schedule determined by the concurrent execution is irrelevant to the expected run length.

As in the sequential algorithm, our $\del$ and $\lookup$ operations never occupy a cell that was previously empty, and therefore they can never extend an existing run or cause two adjacent runs to merge; only $\insr$ operations can increase the run length.
By Proposition~\ref{prop:maptoseq}, provided there are no concurrent $\insr$ operations targeting the same key, the sequential analysis suffices to determine the expected run length.
Consequently, the expected length of any scan over consecutive non-empty cells in the concurrent execution is also $O(x^2)$.

We define a \emph{$c$-bounded fixed-workload scheduler} as an adversary that manages a batch of operations that is fixed in advance,
before the hash function is chosen.
The scheduler is worst-case, and determines the invocation order dynamically: whenever a process is free, the scheduler may choose \emph{any} of the remaining operations in the batch to run next,
subject to
(a) a maximum \emph{point contention} of $c$ per key, that is, the maximum number of concurrent operations targeting the same key at any point in time is at most $c$;
and
(b) a maximum of one $\insr$ operation active on any given key at any time.
Under this scheduling model, we establish the following bound on the step complexity, applicable to both the $\LL/\SC$ and $\CAS$ versions of our algorithm.

\begin{restatable}{theorem}{Amortized}
\label{thm:amortized}
Let $m$ be the size of the table, let $x>1$ be a real number,
and fix any batch of operations that includes at most $(1-{1}/{x})m$ $\insr$ operations.
Under a $c$-bounded fixed-workload scheduler, the expected amortized step complexity per operation in both the $\LL/\SC$ and $\CAS$ versions is $O(x^2 + c)$.
\end{restatable}

\begin{proof}
    We design a \emph{charging scheme} that charges each step taken by process $p$ to some operation, either its own operation or a different one.
    Recall that because concurrent insert operations operate on distinct keys, the state $\tup{\tup{\ast,\ast},\flagcell}$ (required in the $\CAS$ version for colliding insertions) is never written. Consequently, the analyses of the $\LL/\SC$ and $\CAS$ implementations coincide, and we treat them together.

\begin{description}
    \item[Lookup operations:] 
    A lookup operation performs at most one forward and one backward scan, each traversing expected $O(x^2)$ cells, as established earlier. 
    Each step of a scan takes a constant number of read and modify instructions, which are charged to the lookup itself.

    \item[Delete operations:]
    A delete operation performs at most one forward and one backward scan. However, it may execute a super-constant number of instructions on a single cell if it repeatedly encounters a value $\tup{v,\flaginsert}$ or $\tup{v,\flagrestart}$ and fails to update it.

        \begin{itemize}
            \item 
            The initial read of a cell and any successful modify instructions are charged to the delete operation itself.

            \item If the delete operation reads a value that does not contain the key $v$ (or contains $\tup{v,\flagmember}$), it stops re-reading the cell and proceeds. This read is charged to the delete operation itself.

            \item Any subsequent pair of modify and read instructions representing a failed attempt to update the cell is charged to the \emph{interfering} insert or lookup operation that wrote the value triggering the retry.
            \begin{enumerate}
                \item If the value is $\tup{v,\flagrestart}$, the pair is charged to the operation that wrote this value.
                \item If the value is $\tup{v,\flaginsert}$ and it is the first such value written by the current insert operation, the pair is charged to that insert operation.
                \item If the value is $\tup{v,\flaginsert}$ but it overwrites a $\tup{v,\flagrestart}$, the pair is charged to the operation that wrote the preceding $\tup{v,\flagrestart}$.
            \end{enumerate}
        \end{itemize}
        
        A delete operation is charged expected $O(x^2)$ of its own steps. Steps are charged to other operations only when the delete operation re-reads a cell. Consequently, the operation that triggers this retry must be concurrent with the delete operation.
        An insert or lookup operation restarts at most one other operation (by writing $\tup{v,\flagrestart}$ exactly once). Since there are at most $c$ concurrent operations targeting key $v$, at most $c-1$ delete operations can read this $\tup{v,\flagrestart}$ and charge the writer.
        Similarly, if $\tup{v,\flagrestart}$ is subsequently overwritten by $\tup{v,\flaginsert}$, this write can cause at most $c-1$ additional concurrent delete operations to restart.
        Separately, an insert operation that writes $\tup{v,\flaginsert}$ is charged with the retries of at most $c-1$ concurrent delete operations that read this value.
        Thus, an insert or lookup is charged with at most $O(c)$ steps from concurrent delete operations.

    \item[Insert operations:]
    An insert operation begins by searching for the key by performing at most one forward and one backward scan. If found, it terminates. If not, it performs a forward scan to find an empty slot. Since the table is not full, a slot is guaranteed to be found.
        This initial phase comprises at most three scans. With constant work per cell, the expected cost is $O(x^2)$, charged to the insert operation itself.

        In the second phase, the insert repeatedly performs forward scans to finalize the insertion.
        If an insert operation attempts to delete its own copy, it succeeds in modifying the cell after a finite constant number of attempts, unless the cell is marked with $\flagrestart$. This holds because 
        once the operation’s cell transitions to a cleanup state ($\deleted$ or $\collided$), only the owner can modify it. 
        
        A forward scan is repeated only if a different insert or lookup restarts this operation.
        \begin{itemize}
            \item The first forward scan is charged to the insert operation itself.
            \item Any subsequent scans are charged to the operation that triggered the restart.
        \end{itemize}
        Since an insert or lookup restarts at most one other operation, it is charged with at most one additional forward scan of a different insert operation, adding expected $O(x^2)$ steps.
    \end{description}

    Each operation is charged with expected $O(x^2)$ steps for its own execution, at most $O(c)$ steps from interfering deletes, and at most $O(x^2)$ steps from a restarted insert.
    The total expected amortized step complexity is therefore $O(x^2 + c)$.
\end{proof}

\section{Summary and Discussion}

We present a lock-free linear probing hash table that minimizes memory overhead:
it uses only 2 bits of metadata per entry with $\LL$/$\SC$, and $O(\min(\log m, \log n))$ bits using $\CAS$.
Our hash table supports wait-free lookups,
and in the relatively uncontended scenario where keys are not inserted more than once at the same time,
it has roughly the same performance as sequential linear probing,
up to an additive term of the point contention on any given key.

In the world of \emph{sequential} data structures,
much effort is devoted to designing \emph{succinct data structures},
whose memory footprint matches the information-theoretic lower bound on the abstract data type they represent, up to lower-order terms.
For hash tables (or rather, for dictionaries) such constructions have been known for over two decades,
and the precise tradeoff between memory and query time has recently been completely resolved~\cite{LiLYZ23,BFKKL22}.
For concurrent dictionaries, however, no succinct lock-free construction is known; to our knowledge, our construction 
appears to come closest,
at $m(\lceil \log(U) \rceil + 2)$ bits in total (in the $\LL$/$\SC$ version).
Obtaining succinct lock-free dictionaries is an interesting open problem.

\paragraph*{Acknowledgments:}
Hagit Attiya is supported by the Israel Science Foundation 22/1425 and 25/1849. 
Rotem Oshman and Noa Schiller are supported by
NSF-BSF 2022699.

\bibliographystyle{plainurl}
\bibliography{cite}

@article{GaoGH05,
author = {Hui Gao and
                  Jan Friso Groote and
                  Wim H. Hesselink},
title = {Lock-free dynamic hash tables with open addressing},
year = {2005},
issue_date = {July 2005},
publisher = {Springer-Verlag},
address = {Berlin, Heidelberg},
volume = {18},
number = {1},
issn = {0178-2770},
url = {https://doi.org/10.1007/s00446-004-0115-2},
doi = {10.1007/s00446-004-0115-2},
journal = {Distrib. Comput.},
month = jul,
pages = {21-–42},
numpages = {22}
}

@article{MaierSD19,
author = {Maier, Tobias and Sanders, Peter and Dementiev, Roman},
title = {Concurrent Hash Tables: Fast and General(?)!},
year = {2019},
issue_date = {December 2018},
publisher = {Association for Computing Machinery},
address = {New York, NY, USA},
volume = {5},
number = {4},
issn = {2329-4949},
url = {https://doi.org/10.1145/3309206},
doi = {10.1145/3309206},
journal = {ACM Trans. Parallel Comput.},
month = feb,
articleno = {16},
numpages = {32},
}

@book{HerlihyShavitBook2nd,
title = "The Art of Multiprocessor Programming, Second Edition",
author = "Maurice Herlihy and Nir Shavit and Victor Luchangco and Michael Spear",
year = "2020",
doi = "10.1016/C2011-0-06993-4",
isbn = "9780123914064",
publisher = "Elsevier",
}

@inproceedings{PurcellHarris05,
author = {Purcell, Chris and Harris, Tim},
title = {Non-blocking hashtables with open addressing},
year = {2005},
isbn = {3540291636},
publisher = {Springer-Verlag},
address = {Berlin, Heidelberg},
url = {https://doi.org/10.1007/11561927_10},
doi = {10.1007/11561927_10},
booktitle = {Proceedings of the 19th International Conference on Distributed Computing (DISC)},
pages = {108–-121},
numpages = {14},
location = {Cracow, Poland},
}

@article{ShalevShavit06,
  author  = {Shalev, Ori and Shavit, Nir},
  title   = {Split-Ordered Lists: Lock-Free Extensible Hash Tables},
  journal = {Journal of the ACM},
  volume  = {53},
  number  = {3},
  pages   = {379--405},
  year    = {2006}
}

@inproceedings{Michael02,
  author    = {Michael, Maged M.},
  title     = {High Performance Dynamic Lock-Free Hash Tables and List-Based Sets},
  booktitle = {Proceedings of the Fourteenth Annual ACM Symposium on Parallel Algorithms and Architectures (SPAA)},
  year      = {2002},
  pages     = {73--82},
  publisher = {ACM}
}

@inproceedings{ShunBlellochSPAA14,
author = {Shun, Julian and Blelloch, Guy E.},
title = {Phase-Concurrent Hash Tables for Determinism},
year = {2014},
isbn = {9781450328210},
publisher = {Association for Computing Machinery},
address = {New York, NY, USA},
url = {https://doi.org/10.1145/2612669.2612687},
doi = {10.1145/2612669.2612687},
booktitle = {Proceedings of the 26th ACM Symposium on Parallelism in Algorithms and Architectures (SPAA)},
pages = {96-–107},
numpages = {12},
location = {Prague, Czech Republic},
}

@inproceedings{HIHashTableSTOC25,
author = {Attiya, Hagit and Bender, Michael A. and Farach-Colton, Mart\'{\i}n and Oshman, Rotem and Schiller, Noa},
title = {History-Independent Concurrent Hash Tables},
year = {2025},
isbn = {9798400715105},
publisher = {Association for Computing Machinery},
address = {New York, NY, USA},
url = {https://doi.org/10.1145/3717823.3718283},
doi = {10.1145/3717823.3718283},
booktitle = {Proceedings of the 57th Annual ACM Symposium on Theory of Computing (STOC)},
pages = {1283-–1294},
numpages = {12},
location = {Prague, Czechia},
}

@unpublished{Knuth1963OpenAddressing,
  author       = {Donald E. Knuth},
  title        = {Notes on ``Open'' Addressing},
  year         = {1963},
  month        = {July},
  note         = {Unpublished memorandum},
}

@inproceedings{HerlihyShavitTzafrir08,
author = {Herlihy, Maurice and Shavit, Nir and Tzafrir, Moran},
title = {Hopscotch Hashing},
year = {2008},
isbn = {9783540877783},
publisher = {Springer-Verlag},
address = {Berlin, Heidelberg},
booktitle = {Proceedings of the 22nd International Symposium on Distributed Computing (DISC)},
pages = {350-–364},
numpages = {15},
location = {Arcachon, France},
doi = {10.1007/978-3-540-87779-0_24},
}

@misc{Click2007,
    title = {A Lock-Free Wait-Free Hash Table},
    url = {http://www.stanford.edu/class/ee380/Abstracts/070221},
    author = {Cliff Click},
    year = {2007},
    note = {Accessed on 2024-07-04}
}

@INPROCEEDINGS{LockFreeCuckoo,
  author={Nguyen, Nhan and Tsigas, Philippas},
  booktitle={2014 IEEE 34th International Conference on Distributed Computing Systems}, 
  title={Lock-Free Cuckoo Hashing}, 
  year={2014},
  volume={},
  number={},
  pages={627-636},
 }

@inproceedings{LiLYZ23,
  author={Li, Tianxiao and Liang, Jingxun and Yu, Huacheng and Zhou, Renfei},
  booktitle={2023 IEEE 64th Annual Symposium on Foundations of Computer Science (FOCS)}, 
  title={Tight Cell-Probe Lower Bounds for Dynamic Succinct Dictionaries}, 
  year={2023},
  volume={},
  number={},
  pages={1842-1862},
  doi={10.1109/FOCS57990.2023.00112}}

@inproceedings{BFKKL22,
author = {Bender, Michael A. and Farach-Colton, Mart\'{\i}n and Kuszmaul, John and Kuszmaul, William and Liu, Mingmou},
title = {On the optimal time/space tradeoff for hash tables},
year = {2022},
isbn = {9781450392648},
publisher = {Association for Computing Machinery},
address = {New York, NY, USA},
url = {https://doi.org/10.1145/3519935.3519969},
doi = {10.1145/3519935.3519969},
booktitle = {Proceedings of the 54th Annual ACM SIGACT Symposium on Theory of Computing (STOC)},
pages = {1284--1297},
numpages = {14},
location = {Rome, Italy},
}

@InProceedings{blellochWei20,
  author =	{Blelloch, Guy E. and Wei, Yuanhao},
  title =	{{LL/SC and Atomic Copy: Constant Time, Space Efficient Implementations Using Only Pointer-Width CAS}},
  booktitle =	{34th International Symposium on Distributed Computing (DISC 2020)},
  pages =	{5:1--5:17},
  series =	{Leibniz International Proceedings in Informatics (LIPIcs)},
  ISBN =	{978-3-95977-168-9},
  ISSN =	{1868-8969},
  year =	{2020},
  volume =	{179},
  editor =	{Attiya, Hagit},
  publisher =	{Schloss Dagstuhl -- Leibniz-Zentrum f{\"u}r Informatik},
  address =	{Dagstuhl, Germany},
  URL =		{https://drops.dagstuhl.de/entities/document/10.4230/LIPIcs.DISC.2020.5},
  URN =		{urn:nbn:de:0030-drops-130831},
  doi =		{10.4230/LIPIcs.DISC.2020.5}
}

@inproceedings{Moir97,
author = {Moir, Mark},
title = {Practical implementations of non-blocking synchronization primitives},
year = {1997},
isbn = {0897919521},
publisher = {Association for Computing Machinery},
address = {New York, NY, USA},
url = {https://doi.org/10.1145/259380.259442},
doi = {10.1145/259380.259442},
booktitle = {Proceedings of the Sixteenth Annual ACM Symposium on Principles of Distributed Computing},
pages = {219-–228},
numpages = {10},
location = {Santa Barbara, California, USA},
series = {PODC '97}
}

@inproceedings{AndersonMoir95,
author = {Anderson, James H. and Moir, Mark},
title = {Universal constructions for multi-object operations},
year = {1995},
isbn = {0897917103},
publisher = {Association for Computing Machinery},
address = {New York, NY, USA},
url = {https://doi.org/10.1145/224964.224985},
doi = {10.1145/224964.224985},
booktitle = {Proceedings of the Fourteenth Annual ACM Symposium on Principles of Distributed Computing},
pages = {184-–193},
numpages = {10},
location = {Ottowa, Ontario, Canada},
series = {PODC '95}
}

@inproceedings{JayantiPetrovic03,
author = {Jayanti, Prasad and Petrovic, Srdjan},
title = {Efficient and practical constructions of LL/SC variables},
year = {2003},
isbn = {1581137087},
publisher = {Association for Computing Machinery},
address = {New York, NY, USA},
url = {https://doi.org/10.1145/872035.872078},
doi = {10.1145/872035.872078},
booktitle = {Proceedings of the Twenty-Second Annual Symposium on Principles of Distributed Computing},
pages = {285-–294},
numpages = {10},
location = {Boston, Massachusetts},
series = {PODC '03}
}

@inproceedings{IsraeliRappoportPODC94,
author = {Israeli, Amos and Rappoport, Lihu},
title = {Disjoint-Access-Parallel Implementations of Strong Shared Memory Primitives},
year = {1994},
isbn = {0897916549},
publisher = {Association for Computing Machinery},
address = {New York, NY, USA},
url = {https://doi.org/10.1145/197917.198079},
doi = {10.1145/197917.198079},
booktitle = {Proceedings of the Thirteenth Annual ACM Symposium on Principles of Distributed Computing},
pages = {151–-160},
numpages = {10},
location = {Los Angeles, California, USA},
series = {PODC '94}
}

@manual{arm_ddi0487_mb,
  title        = {{Arm Architecture Reference Manual for A-profile architecture}},
  author       = {{Arm Ltd.}},
  organization = {Arm Ltd.},
  note         = {Document ID: DDI 0487, Issue M.b},
  year         = {2026}, 
  url          = {https://developer.arm.com/documentation/ddi0487/}
}

\end{document}